\documentclass[acmsmall,authorversion,nonacm]{acmart}

\usepackage{hyperref}
\usepackage[nodisplayskipstretch]{setspace}
\usepackage{booktabs}
\usepackage{pdflscape}
\usepackage{adjustbox}
\usepackage{tcolorbox}
\usepackage{xcolor}
\usepackage{tikz}
\usepackage{mathrsfs}
\usetikzlibrary{arrows}
\usetikzlibrary{snakes}
\usetikzlibrary{shapes.geometric}
\usetikzlibrary{shapes}
\usetikzlibrary{positioning}
\usetikzlibrary{patterns}
\usepackage{tkz-graph}
\usepackage{graphicx}
\usepackage{enumitem}
\usepackage{algpseudocode}
\usepackage{algorithm}
\usepackage{subcaption}
\usepackage[toc]{appendix}
\usepackage{accents}

\tikzset{set/.style={draw,circle,inner sep=0pt,align=center}}
\algnewcommand{\IfThen}[2]{% \IfThenElse{<if>}{<then>}{<else>} :: one line if-then
  \State \algorithmicif\ #1\ \algorithmicthen\ #2}
\algnewcommand{\IfThenElse}[3]{% \IfThenElse{<if>}{<then>}{<else>} :: one line if-then-else
  \State \algorithmicif\ #1\ \algorithmicthen\ #2\ \algorithmicelse\ #3}
\algdef{SE}[DOWHILE]{Do}{DoWhile}{\algorithmicdo}[1]{\algorithmicwhile\ #1}%

\AtBeginDocument{%
  \providecommand\BibTeX{{%
    \normalfont B\kern-0.5em{\scshape i\kern-0.25em b}\kern-0.8em\TeX}}}

\setcopyright{acmcopyright}
\copyrightyear{2023}
\acmYear{2023}
\acmDOI{XXXXXXX.XXXXXXX}

\begin{document}

\title{Models and algorithms for simple disjunctive temporal problems}

\author{Carlo S. Sartori}
\email{carlo.sartori@cs.kuleuven.be}
\orcid{0000-0003-2140-2925}
\author{Pieter Smet}
\email{pieter.smet@cs.kuleuven.be}
\orcid{0000-0002-3955-7725}
\author{Greet Vanden Berghe}
\email{greet.vandenberghe@cs.kuleuven.be}
\orcid{0000-0002-0275-5568}
\affiliation{%
  \institution{\\KU Leuven, Department of Computer Science}
  \streetaddress{Gebroeders De Smetstraat 1}
  \city{Gent}
  \country{Belgium}
  \postcode{9000}
}

%\renewcommand{\shortauthors}{Sartori et al.}

\begin{abstract}
  Simple temporal problems represent a powerful class of models capable of describing the
  temporal relations between events that arise in many real-world
  applications such as logistics, robot planning and management systems. The classic simple temporal problem
  permits each event to have only a single release and due date. In this paper, we focus on
  the case where events may have an arbitrarily large number of release and due dates. This type of problem, however, has been referred to by various names. In order to simplify and standardize nomenclatures, we introduce the name Simple Disjunctive Temporal Problem. We provide three mathematical models to describe this problem using constraint programming and linear programming. To efficiently solve simple disjunctive temporal problems, we design two new algorithms inspired by previous research, both of which exploit the problem's structure to significantly reduce their space complexity.
  Additionally, we implement algorithms from the literature and provide the first in-depth empirical study comparing methods to
  solve simple disjunctive temporal problems across a wide range of experiments. Our analysis and conclusions offer guidance for future researchers and practitioners when tackling similar temporal constraint problems in new applications. All results, source code and instances are made publicly available to further assist future research.
\end{abstract}

%% For standardization, to simplify naming, (after large number of domains.)




\maketitle

\section{Introduction}

Simple Temporal Problems (STPs) provide a formal structure to describe possible time relations between events. These time relations feature in a wide variety of real-world problems \cite{art:time21} and include precedences, maximum elapsed time, and a single release and due date per event. A major advantage of STPs is that they are solvable in polynomial time with standard shortest path methods~\cite{art:stp}. Nevertheless, researchers and practitioners are still fairly limited with respect to what can be modeled using STPs. Even an otherwise simple feature such as multiple release and due dates per event cannot be expressed with STPs.

Alternatively, Disjunctive Temporal Problems (DTPs) offer a much broader framework for describing time relations. However, this expressiveness is offset by the fact that they usually represent NP-Complete problems \cite{art:dtps}. Simple Disjunctive Temporal Problems (SDTPs) are a primitive type of DTP which generalize STPs and retain efficient polynomial-time solution methods. SDTPs extend STPs by enabling multiple, non-overlapping release and due dates per event.

Throughout the academic literature, SDTPs are referred to by many names: \textsc{Star} class of problems, zero-one extended STPs and t\textsubscript{2}DTPs. While it is difficult to know for sure why exactly so much terminology exists for the same problem, one could speculate that it might be that different researchers have each arrived at SDTPs from different angles. Although the particular reason for this terminological variance is not our main concern, it is clear that it results in a highly inefficient scenario. Researchers and practitioners alike end up being held back by the burdensome task of needing to discover what is known about SDTPs when the findings are catalogued under different names. Moreover, even when one does locate existing literature concerning SDTPs, those research papers are primarily theoretical in nature and provide neither a practical implementation of the methods nor empirical insights concerning how the behavior of those algorithms compares under different scenarios.

We have experienced precisely the situation outlined when trying to compare different approaches for scheduling tasks with multiple time windows in logistical problems such as vehicle routing with synchronizations \cite{art:vrpms,art:delsynch}, pickup and delivery with transshipment \cite{art:pdpt}, dial-a-ride with transfers \cite{art:darpt} and truck driver scheduling with interdependent routes \cite{art:ssvb-1}. SDTPs are an excellent model for scheduling in these problems. Nevertheless, one needs extremely efficient SDTP methods when solving these logistical problems via local search heuristics, for example. Given the fact that the literature is not only difficult to navigate but also lacks empirical results, we needed to (i) find existing methods, (ii) implement them and (iii) evaluate their advantages and limitations in different cases before employing SDTPs in our applications.

The goal of this paper is therefore fourfold. First, we propose a standard nomenclature to refer to SDTPs so that future researchers can refer to the same problem by the same name. Second, we explore existing methods and develop new algorithms to solve SDTPs that are capable of not only reducing the theoretical asymptotic worst-case time and space complexities but also ensuring good performance in practice. Third, we provide an empirical study alongside open source implementations of all of the techniques and also make our instances publicly available with the aim of avoiding duplicate work\footnote{The complete code repository will be made available at a later date.}. Fourth and finally, we hope this paper will serve as a foundation for researchers and practitioners who would like to apply SDTPs in their work and build upon our research and implementations to easily achieve their goals.

\section{The simple disjunctive temporal problem} \label{sec:sdtp}

Let us begin by formally defining simple temporal problems, which will be useful when introducing simple disjunctive temporal problems. 

\begin{definition}[Simple Temporal Problem \cite{art:stp}]\label{def:stp}
A \textit{Simple Temporal Network} (STN) is denoted $N=(T,C)$ where $T$ is the set of \textit{variables} or \textit{time-points} and $C$ is the set of binary constraints relating variables of $T$. A time-point $i \in T$ has a closed domain $[l_i,u_i],\ l_i,u_i \in \mathbb{R}$. Constraints in $C$ are \textit{simple temporal constraints} given as a tuple $(i,j,w_{ij}) \in C,\ i,j \in T,\  w_{ij} \in \mathbb{R}$ which corresponds to Equation~\ref{eq:stc}:
\begin{equation}
    s_i - s_j \leq w_{ij} \label{eq:stc}
\end{equation}
\noindent where $s_i,s_j \in \mathbb{R}$ denote the solution values assigned to time-points $i$ and $j$. 

 The STP involves determining whether its associated STN is consistent. A network $N$ is consistent iff a \textit{feasible schedule} or \textit{solution} $s$ can be derived such that times $s_i$ assigned to each $i \in T$ respect all constraints present in $C$ and the domains of each time-point.
\end{definition}

 \citet{art:stp} showed that an STN can be represented with a distance graph where time-points $T$ are nodes and constraints $C$ are arcs connecting these nodes. First, let us associate a special time-point $\alpha$ with domain $[0,0]$ to represent the beginning of the time horizon $s_\alpha = 0$. This fixed point can be used to write unary constraints such as domain boundaries over time-points in $T$ as simple temporal constraints. For example, the bound $[l_i,u_i]$ for $i \in T$ can be written as:
 \begin{align*}
     s_\alpha - s_i &\leq -l_i\\
     s_i - s_\alpha &\leq u_i
 \end{align*}
 
 We can then associate two distance graphs with STN $N=(T,C)$: the \textit{direct graph} $G_D=(V,A_D)$ and the \textit{reverse graph} $G_R=(V,A_R)$. For both of these graphs $V=T \cup \{\alpha\}$. Arc set $A_D= C \cup \{(\alpha,i,-l_i), (i,\alpha,u_i)\ \forall\ i \in T\}$ for which an element $(i,j, w_{ij}) \in A_D$ denotes an arc from node $i$ to $j$ with weight $w_{ij}$. Meanwhile, $A_R$ is the same as $A_D$ but with the direction of each arc reversed: $(i,j,w_{ij}) \in A_D$ is $(j,i,w_{ij}) \in A_R$.

Determining consistency of an STN then reduces to verifying the existence of negative-cost cycles in either $G_D$ or $G_R$. If there is no negative-cost cycle, the shortest path distance $\tau_{\alpha i}$ from node $\alpha \in V$ to every other node $i \in V \backslash \{\alpha\}$ provides a feasible schedule. When computed over $G_D$, $s_i = -\tau_{\alpha i}$ provides the \textit{earliest feasible schedule}. Meanwhile, when computed over $G_R$, $s_i=\tau_{\alpha i}$ provides the \textit{latest feasible schedule}. The earliest feasible schedule can be defined as the solution $s$ for which given any other feasible solution $s'$ to the SDTP, it holds that $s_i \leq s'_i,\ \forall\ i \in T$. Similarly, for the latest feasible schedule it holds that $s_i \geq s'_i,\ \forall\ i \in T$.

There are many algorithms that can be used to detect negative-cost cycles quickly in a distance graph \cite{art:sp-fp}. One of the most simple is \textsc{BellmanFord} \cite{book:cormen}. Indeed, this is an algorithm employed by most methods to solve SDTPs. For the remainder of this paper, we always consider \textsc{BellmanFord} to refer to its implementation as a label-correcting algorithm using a first-in, first-out queue \cite{art:sp-fp}.

STPs can only accommodate time-points for which the domain is a single interval. In order to tackle problems where time-points may be assigned values in one of several disjunctive intervals, a more expressive model is required. Let us now turn our attention to the main problem in this paper: the simple disjunctive temporal problem.

\begin{definition}[Simple Disjunctive Temporal Problem]\label{def:sdtp}
A \textit{Simple Disjunctive Temporal Network} (SDTN) is denoted $N=(T,C)$, where $T$ is the set of time-points and $C=C_1 \cup C_2$ are the constraints over these time-points. Based on the classification introduced by \citet{art:kra}, constraints in $C_1$ are Type 1 while those in $C_2$ are Type 2.
\begin{enumerate}
    \item[] (Type 1) Simple temporal constraints $(i,j,w_{ij}) \in C_1,\ i,j \in T$ representing Equation \ref{eq:stc}
    \item[] (Type 2) Simple disjunctive constraints $(i,D_i) \in C_2$, where $i \in T$ and $D_i$ is a list of \textit{intervals} or \textit{domains} denoted $[l^c_i,u^c_i] \in D_i,\ l^c_i,u^c_i \in \mathbb{R}$ representing the disjunction
    \begin{equation*}
     \bigvee_{[l^c_i,u^c_i] \in D_i} (l^c_i \leq s_i \leq u^c_i)   
    \end{equation*}
    Note that $C_2$ includes unary constraints  ($|D_i|=1$). $T_D \subseteq T$ denotes the set of all time-points for which $|D_i| > 1$. 
\end{enumerate}

In order to solve the SDTP, we need to determine whether SDTN $N$ is consistent. Similar to STPs, $N$ is consistent iff a feasible solution $s$ can be derived which respects both constraints $C_1$ and $C_2$. 

\end{definition}

We assume that domains in $D_i$ are sorted in ascending order \cite{art:kra,art:cra}. Let $K = \max_{(i,D_i) \in C_2} |D_i|$ denote the largest number of domains for any given time-point and $\omega = \sum_{(i,D_i) \in C_2} |D_i|$ denote the total number of domains in the instance. Let us further denote by $L(D_i)$ and $U(D_i)$ the lower and upper bound values in all domains of $i \in T$, respectively. The \textit{global boundary} of $i$ is given by $[L(D_i),U(D_i)]$ such that $s_i$ must belong to this boundary. However, some values within these bounds can still be infeasible. In other words: the domains of time-points are not continuous.

The existence of Type 2 constraints means we cannot solve the problem directly via shortest paths. Nevertheless, we can use the global boundaries of the time-points to redefine graphs $G_D$ and $G_R$ with $A_D=C_1 \cup \{(\alpha,i,-L(D_i)), (i,\alpha,U(D_i))\ \forall\ (i,D_i) \in C_2\}$ and $A_R$ ($A_D$ with all arc directions reversed). These graphs can be used to compute \textit{lower-} and \textit{upper-bound solutions} for the SDTP while employing shortest path algorithms in the same way as for STPs. If a negative cycle exists in $G_D$ or $G_R$ when considering these global boundaries, then the associated SDTP instance is definitely infeasible.

Once a domain $d_i \in D_i$ has been selected for each time-point $(i,D_i) \in C_2$, the SDTP reduces to an STP. Indeed, some of the special-purpose algorithms available in the literature exploit this problem structure to solve SDTP instances. Section \ref{sec:algs} will discuss this further.

\subsection{Related problems and classification}

As noted in this paper's introduction, the SDTP has been referred to by various names in the literature. \citet{art:ult} term it the \textsc{Star} class of problems because the multiple domains per time-point create connections to the beginning of time $\alpha$, which resembles the shape of a star. \citet{art:kumar-estp} refers to the problem as zero-one extended STPs, where subintervals of a time-point's domain are associated with a weight that is either 0 when the interval is infeasible, or~1 when the interval is feasible. An SDTP solution has therefore been constructed when the sum of the weights of selected intervals is $|C_2|$. Meanwhile, \citet{art:kra} introduced Restricted Disjunctive Temporal Problems (RDTPs) which contain constraints of Type 1, 2 and 3\footnote{Type 3 constraints consider two different time-points $i,j \in T,\ i\neq j$ and relate them via a disjunction with exactly two terms in the form $(l'_i   \leq s_i \leq u'_i) \lor (l'_j \leq s_j \leq u'_j)$, where $l'_i,u'_i,l'_j,u'_j \in \mathbb{R}$ denote bounds for time-points $i$ or $j$. This type of constraint is not handled in this paper but interested readers are referred to \cite{art:kra,art:cra} for more information about them.}. The SDTP therefore arises as a special case of RDTPs when there are no Type 3 constraints. \citet{art:cra} also refer to SDTPs as t\textsubscript{2}DTPs, framing them as DTPs that only contain constraints of Type 1 and 2. 

In a move to simplify and unify nomenclatures, we have decided to introduce the name \textit{Simple Disjunctive Temporal Problem} following the same reasoning behind the naming of STPs \cite{art:stp}. Figure~\ref{fig:tcps-class} below situates the SDTP within the larger scheme of DTPs.

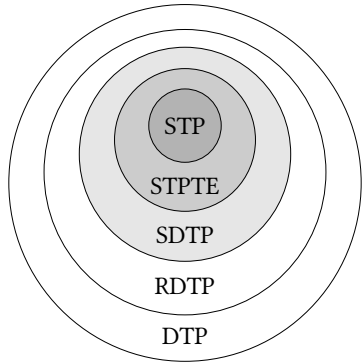
\begin{figure}[h]
  \begin{center}
  \resizebox{0.35\linewidth}{!}{
    \begin{tikzpicture}
    \node[set,text width=5cm,label={[below=125pt of dtp,text opacity=1]DTP}] 
      at (0,-0.8) (dtp) {};
      \node[set,text width=4cm,label={[below=95pt of dtp,text opacity=1]RDTP}] 
      at (0,-0.65) (rdtp) {};
      \node[set,fill=gray!20,text width=3cm,label={[below=68pt of rdtp,text opacity=1]SDTP}] 
      at (0,-0.4) (sdtp) {};
      \node[set,fill=gray!40,text width=2cm,label={[below=40pt of sdtp]STPTE}] 
      (stpt) at (0,-0.2)  {};
      \node[set,fill=gray!60,text width=1cm] (stp) at (0,0) {STP};
    \end{tikzpicture}
    }
    \caption{Classification of DTPs in a set diagram.}
        \label{fig:tcps-class}
  \end{center}
\end{figure}

SDTPs generalize STPs since the latter can be cast as an SDTP for which $|D_i|=1,\ \forall\ i \in T$. They also generalize the Simple Temporal Problem with Taboo regions featuring both instantaneous Events and processes of constant duration (STPTE) \cite{art:stpts}. This class of problems differs from SDTPs because STPTEs define common intervals when no time-point can be scheduled rather than individual intervals per time-point. This clearly demonstrates how SDTPs generalize STPTEs. However, when the duration of processes can vary within an interval in STPs with taboo regions, then SDTPs cannot generalize them because Type 3 constraints are needed \cite{art:stpts}. Finally, RDTPs generalize all of the aforementioned problems while DTPs further generalize RDTPs. The gray area in Figure \ref{fig:tcps-class} represents the problems that the models and algorithms in this paper address. 

Another problem related to SDTPs is the time-dependent STP \cite{art:td-stp}. While it might not be an obvious connection at first, in the time-dependent version of STPs the weight $w_{ij}$ in Type 1 constraints is not a constant but rather a function $f(s_i,s_j)$ which depends on the values assigned to the time-points. When $f(s_i,s_j)$ is a piecewise linear and partial function over the global boundary $[L(D_j),U(D_j)]$, it is possible to cast the SDTP as a time-dependent STP. In this case, the function is defined between $\alpha$ and every $j \in T$, that is $f(s_\alpha,s_j)$. The pieces of function $f(s_\alpha,s_j)$ represent the domains of time-point $j$. This relation has not previously been established in the literature and one of the possible reasons could be that \citet{art:td-stp} originally focused more on total functions given that each time-point had a single domain in their application. The effect of partial functions in the development of algorithms will be discussed further in Section \ref{sec:algs}. We opted not to include the time-dependent STP in Figure \ref{fig:tcps-class} so as to maintain a clear relation between problems that are often considered together in the temporal reasoning literature, namely those that deal with disjunctions. Nevertheless, the connection we have established has important implications for computing solutions to SDTPs.

\subsection{Constraint programming model}

Constraint Programming (CP) tools are widely used in planning and scheduling domains. Hence, it is worth considering whether CP is a good candidate for solving SDTPs in practice. The corresponding CP model is: 
\begin{align}
  s_{i} - s_{j} \leq w_{ij}, &\quad \forall\ (i,j,w_{ij}) \in C_1 \label{cp:1}\\
  \bigvee_{[l^k_i,u^k_i] \in D_i} (l^k_i \leq s_i \leq u^k_i), &\quad \forall\ (i,D_i) \in C_{2} \label{cp:2}
\end{align}
\noindent which is actually the same set of equations as those in Definition \ref{def:sdtp}. This is very convenient because we essentially have a one-to-one mapping between the classic definition of SDTPs and their CP formulation. For simplicity, we will refer to model (\ref{cp:1})-(\ref{cp:2}) as CP.

A simplified CP formulation can be written as follows:
\begin{align}
  s_{i} - s_{j} \leq w_{ij}, &\quad \forall\ (i,j,w_{ij}) \in C_1\label{cps:1}\\
  L(D_i) \leq s_i \leq U(D_i), &\quad \forall\ (i,D_i) \in C_2\label{cps:2}\\
  s_i \notin \Phi_i, &\quad \forall\ i \in T\label{cps:3}
\end{align}

Constraints (\ref{cps:1}) are the same as Equation \ref{eq:stc}, while Constraints (\ref{cps:2}) model the global boundaries of time-points. Constraints (\ref{cps:3}) are the \textit{compatibility constraints} and serve as a replacement for disjunctive Constraints (\ref{cp:2}). Set $\Phi_i$ contains the enumeration of all infeasible assignments to $i \in T$ that belong to the interval $[L(D_i),U(D_i)]$. In other words: $\Phi_i = \{u^1_i+1,u^1_i+2,\dots,  l^2_i-1,u^2_i+1,\dots,l^k_i-1\},\ k=|D_i|$. Adding these constraints is only possible if we make the additional assumption that $s_i \in \mathbb{Z}$. However, given that typical CP tools only operate with integer variables, this assumption is not necessarily restrictive in practice. We will refer to the formulation defined by (\ref{cps:1})-(\ref{cps:3}) as \textit{Simplified Constraint Programming} (SCP).

\subsection{Integer linear programming model}

 Integer Linear Programming (ILP) is also often used in the planning and scheduling domains, which motivated us to also formulate the SDTP in ILP form. First, let us define the binary decision variable $x^c_i$ which takes value 1 whenever solution value $s_i$ belongs to domain $[l^c_i,u^c_i] \in D_i,\ (i,D_i) \in C_2$ and 0 otherwise. The corresponding ILP model for SDTPs is:
 
\begin{align}
  s_i - s_{j} \leq w_{ij}, &\quad \forall\ (i,j,w_{ij}) \in C_1\label{ilp:1}\\ 
  l^c_i - M^L_i(1-x^c_i) \leq s_i, &\quad \forall\ (i,D_i) \in C_{2},\ [l^c_i,u^c_i] \in D_i \label{ilp:2}\\ 
  s_i \leq u^c_i + M^U_i(1-x^c_i), &\quad \forall\ (i,D_i) \in C_{2},\ [l^c_i,u^c_i] \in D_i\label{ilp:3}\\ 
  \sum_{[l^c_i,u^c_i] \in D_i} x^c_i = 1, &\quad \forall\ (i,D_i) \in C_{2}\label{ilp:4}\\ 
  x^c_i \in \{0,1\}, &\quad \forall\ (i,D_i) \in C_{2},\ [l^c_i,u^c_i] \in D_i \label{ilp:5}
\end{align}

Constraints (\ref{ilp:1}) refer to the simple temporal constraints (Equation \ref{eq:stc}). Meanwhile, Constraints (\ref{ilp:2})-(\ref{ilp:3}) restrict the values assigned to solution $s$ so that they belong to the active bounds defined by variables $x^c_i$. Note that Constraints (\ref{ilp:2})-(\ref{ilp:3}) are big-$M$ constraints. They can be tightened by setting, for each $(i, D_i) \in C_2$:
\begin{align*}
M^L_i &= \max_{[l^c_i,u^c_i]\in D_i} l^c_i - L(D_i)\\
M^U_i &= U(D_i) - \min_{[l^c_i,u^c_i]\in D_i} u^c_i    
\end{align*}
\noindent Constraints (\ref{ilp:4}) ensure that exactly one domain is selected per time-point $i \in T$. Finally, Constraints~(\ref{ilp:5}) restrict $x^c_i$ variables to take binary values. Recall that an SDTP is a feasibility problem, with this explaining why there is no objective function present in this ILP.

\bigskip

All three models (ILP, CP and SCP) can be used to solve SDTPs by employing a state-of-the-art solver such as IBM's CPLEX. However, these solvers are often financially expensive. Furthermore, specific methods can provide guarantees concerning expected run times, such as asymptotic polynomial worst-case time complexity. In the following section, we describe many algorithms that can be used to quickly solve SDTPs in practice.

\section{Algorithms} \label{sec:algs}

A variety of special-purpose algorithms have been proposed for SDTPs. All of the algorithms that are presented in this section will be implemented for our computational experiments. It is worth noting that all of the algorithms provide a guaranteed polynomial asymptotic worst-case time complexity. Algorithms are presented in chronological order of publication date. In some cases, we adapted algorithms to ensure they could be implemented efficiently in practice. For that reason, we try to provide as many implementation details as possible. In all cases where details are missing, we refer interested readers to our code for deeper inspection.

We assume that each algorithm receives as input an SDTP instance containing network $N=(T,C)$ and associated graphs $G_D$ and $G_R$. Some algorithms also receive additional structures, which we detail for the individual method whenever necessary. All algorithms return a solution vector $s$. When the SDTP instance is feasible, each entry $s_i$ contains a time assigned to $i \in T$ which in combination with the other entries renders the solution feasible (network $N$ consistent). Whenever the SDTP instance is infeasible, $s=\emptyset$ is returned.

\subsection{Upper-Lower Tightening}

\citet{art:ult} introduced the \textit{Upper-Lower Tightening} (ULT) algorithm to tackle general disjunctions in DTPs. The original intention behind ULT was to tighten disjunctive constraints and simplify DTP instances. However, \citet{art:ult} were the first to show that SDTPs could be solved in polynomial time by means of ULT.

The ULT algorithm operates with constraints between two variables denoted as an interval. The first step is to therefore define set $H = \{ (i,j)\ \forall\ (i,j,w_{ij}) \in C_1 \} \cup \{ (\alpha,i)\ \forall\ (i,D_i) \in C_2\}$. Let us further assume that $(i,j) \in H \implies (j,i) \notin H$. \textit{Boundary} sets $B_{ij}$ are defined for $(i,j) \in H\ :\ i \neq \alpha$, $B_{ij}=\{[-w'_{ji}, w'_{ij}] \}$ where $w'_{ij} = w_{ij}\text{ if } (i,j,w_{ij}) \in C_1$, otherwise $w'_{ij} = +\infty$, and $w'_{ji} = w_{ji}\text{ if } (j,i,w_{ji}) \in C_1$, otherwise $w'_{ji} = +\infty$. Meanwhile, for $(\alpha,i) \in H$ we relate $i$ to the beginning of time $\alpha$ via $B_{\alpha i}=D_i\ :\  (i,D_i) \in C_{2}$. We will use the notation $L(B_{ij})$ and $U(B_{ij})$ to denote the lower and upper bounds in $B_{ij}$, respectively.

Algorithm \ref{alg:ult} outlines how ULT works. First, a distance matrix $\delta$ is initialized in line 1. The main loop of the algorithm spans lines 2-8. In lines 3-4, some entries of the distance matrix $\delta$ are updated according to the current bounds $B$ of each pair $(i,j) \in H$. \textsc{FloydWarshall} \cite{book:cormen} is then used to update matrix $\delta$ by computing All-Pairs Shortest Paths (APSPs) using the current values in $\delta$ as the arc weights (line 5). A temporary boundary set $B'$ is created in line 6 with the newly computed values in $\delta$. Note that in the implementation itself we do not create $B'$ since we can use matrix $\delta$ directly in its place whenever needed (for example, in line 7). The intersection of $B'$ and $B$ is computed in line 7. Here, we follow the definition of the $\cap$ operation introduced by \citet{art:ult}: it returns a set of intervals whose values are permitted by both $B'$ and $B$. ULT iterates so long as there are changes to the bounds in $B$, denoted by operation $\textsc{Change}$, and no bound is either empty or infeasible. All checks in line 8 can be performed in $O(1)$ time by maintaining the correct flags after lines 6-7. Similarly, lines 3-4 can be performed during operation $\cap$ in line 7 without increasing the asymptotic worst-case time complexity. Lines 9-11 prepare solution $s$ to be returned. If the instance is feasible then line 10 assigns the earliest feasible schedule to $s$, otherwise $\emptyset$ is returned.

\begin{algorithm}[H]
  \caption{ULT}
  \label{alg:ult}
  \footnotesize
  \begin{algorithmic}[1] % The number tells where the line numbering should start
    \State $\delta_{ij} \gets +\infty,\ \forall\ i,j \in T \cup \{\alpha\}$
    \Do
    \State $\delta_{ij} \gets U(B_{ij}),\ \forall\ (i,j) \in H$ \Comment{Update current $\delta$ entries with new bounds}
    \State $\delta_{ji} \gets -L(B_{ij}),\ \forall\ (i,j) \in H$
    \State $\textsc{FloydWarshall}(\delta)$ \Comment{Update distance matrix $\delta$}
    \State $B'_{ij} \gets \{[-\delta_{ji},\delta_{ij}]\},\ \forall\ (i,j) \in H$
    \State $B \gets B \cap B'$ \Comment{Tightens boundaries}
    \DoWhile{$\textsc{Change}(B) \textbf{ and } (B_{ij} \neq \emptyset \textbf{ and } L(B'_{ij}) \leq U(B'_{ij}),\ \forall (i,j) \in H)$}
    \State $s \gets \emptyset$
    \IfThen{$B_{ij} \neq \emptyset \textbf{ and } L(B'_{ij}) \leq U(B'_{ij}),\ \forall (i,j) \in H$}{$s_i \gets L(B_{\alpha i}),\ \forall\ i \in T$}
    \State \Return $s$
  \end{algorithmic}
\end{algorithm}

The asymptotic worst-case time complexity of ULT is $O(|T|^3|C|K + |C|^2K^2)$ \cite{art:ult}, while its space complexity is $O(|T|^2)$ due to distance matrix $\delta$. Despite its apparently high computational complexity, ULT is a polynomial time algorithm. Additionally, \citet{art:ult} noted that even when a problem instance contains multiple disjunctions per constraint between time-points $i,j \in T$, and is therefore not an SDTP instance, ULT may successfully remove sufficient disjunctions to reduce the problem to an SDTP. In this case, ULT is guaranteed to solve the problem exactly. This is the only algorithm in our study capable of such a reduction.

\subsection{Kumar's Algorithm}

\citet{art:kumar-estp} proposed a polynomial time algorithm to solve zero-one extended STPs, which essentially correspond to an SDTP. Algorithm \ref{alg:kra} provides a pseudocode outline of how \textit{Kumar's Algorithm} (KA) works. In line 1, a distance matrix $\delta$ is constructed by computing APSPs over graph $G_R$. This step can detect infeasibilities such as if there exists a negative cycle formed by $C_1$ constraints and global boundaries, in which case $\delta = \emptyset$ is returned. 

Matrix $\delta$ can be computed by employing (i) \textsc{FloydWarshall}, (ii) repeated calls to \textsc{BellmanFord} or (iii) Johnson's Algorithm \cite{book:cormen}. \citet{art:kumar-estp} did not specify which method should be used when computing $\delta$ and therefore we will consider both options (ii) and (iii). Option (i) is disregarded due to its overall poor performance during our preliminary experiments.

\begin{algorithm}
  \caption{KA}
  \label{alg:kra}
  \footnotesize
  \begin{algorithmic}[1] % The number tells where the line numbering should start
    \State $\delta \gets \textsc{ComputeDistanceMatrix}(G_R)$
    \IfThen{$\delta = \emptyset$}{\textbf{return} $\emptyset$}
    \State $G_C \gets \textsc{CreateConflictGraph}(\delta, C_2)$ \Comment{Graph $G_C=(E,A_C)$}
    \IfThen{$G_C = \emptyset$}{\textbf{return} $\emptyset$}
    \State $G_B \gets \textsc{CreateBipartiteGraph}(G_C)$ \Comment{Graph $G_B=(E,E',A_B)$, where $E'$ is a copy of $E$}
    \State $G_F \gets \textsc{SolveMaxFlow}(G_B)$ \Comment{From source $\theta_1$ to sink $\theta_2$, with $G_F$ corresponding to the residual graph}
    \State $S \gets \{(\theta_1,e^c_i)\ :\ e^c_i \notin R(G_F, \theta_1)\} \cup \{(e^{k\prime}_j,\theta_2)\ :\ e^{k\prime}_j \in R(G_F, \theta_1)\}$ \Comment{Minimum cut in $G_F$}
    \State $S' \gets \{e^c_i\ :\ (\theta_1,e^c_i) \in S\ \lor (e^{c\prime}_i,\theta_2) \in S\}$ \Comment{Vertex cover for $G_C$}
    \State $S'' \gets E \backslash S'$
    \IfThen{$|S''| \neq |T|$}{\textbf{return} $\emptyset$}
    \State $\textsc{UpdateGraph}(G_R,S'')$
    \State $s \gets \textsc{BellmanFord}(G_R, \alpha)$ \Comment{Solve STP}
    \State \Return $s$
  \end{algorithmic}
\end{algorithm}

Line 3 proceeds to create a conflict graph $G_C=(E,A_C)$ with the domains from the SDTP. First, set $E$ of intervals is defined as $E = \{ e^c_{i}\ :\ (i,D_i) \in C_2,\ [l^c_i,u^c_i] \in D_i \}$. Hence, every element $e^c_{i} \in E$ represents exactly one domain of a time-point. A domain $[l^c_i,u^c_i]$ has no corresponding element in $E$ if it produces a size-1 conflict, that is, if the following is true:
\begin{equation*}
    \delta_{i\alpha} + u^c_i < 0\ \lor\ \delta_{\alpha i} - l^c_i < 0
\end{equation*}

Once vertex set $E$ has been created, arc set $A_C$ can be defined. An arc $(e^c_{i},e^k_{j}) \in A_C$ denotes a size-2 conflict between two time-point domains $[l^c_i,u^c_i]$ and $[l^k_j,u^k_j]$. Such a conflict occurs whenever:
\begin{equation*}
    u^c_i + \delta_{ij} - l^k_j < 0
\end{equation*}
Note that size-2 conflicts are also defined between domains of the same time-point $i \in T$. There is always a conflict $(e^c_{i},e^{c+1}_{i}) \in A_C$ because $\delta_{ii}=0$ and $u^c_i < l^{c+1}_i$ (recall from Section \ref{sec:sdtp} that domains are in ascending order).

Procedure \textsc{CreateConflictGraph} returns either graph $G_C$ or $\emptyset$. The latter is returned whenever all domains of a time-point $i \in T$ produce size-1 conflicts. In this case no domain associated with $i$ is included in $E$, thereby implying that the SDTP instance is infeasible. Once graph $G_C$ has been constructed, line 5 creates a bipartite graph $G_B=(E,E',A_B)$ by copying every element $e^c_{i} \in E$ to $e^{c\prime}_{i} \in E'$. For each $(e^c_{i},e^k_{j}) \in A_C$ we create an arc $(e^c_{i},e^{k\prime}_{j}) \in A_B$. All arcs in $A_B$ connect an element of $E$ to an element of $E'$.

Line 6 solves a maximum bipartite matching over $G_B$ as a maximum flow problem (max-flow), producing the residual graph $G_F$ \cite{book:cormen}. To solve the problem in the form of a max-flow, we introduce a source node $\theta_1$ and a sink node $\theta_2$ to $G_B$. Arcs $(\theta_1,e^c_{i}),\ \forall\ e^c_{i} \in E$ and $(e^{k\prime}_{j},\theta_2),\ \forall\ e^{k\prime}_{j} \in E'$ are included in the graph together with all arcs in $A_B$. Additionally, all arcs are given unitary capacity. Then, it suffices to solve a max-flow from $\theta_1$ to $\theta_2$ to produce $G_F$. 

The minimum-cut $S$ is computed in $G_F$ thanks to the max-flow min-cut theorem (line 7). $R(G_F,\theta_1)$ denotes the set of nodes that are reachable from source $\theta_1$ in $G_F$ (meaning there is a path with positive residual capacity). Line 8 merges node copies in $S$ to create $S'$, which is a vertex cover for $G_C$ when seen as an undirected graph. Since $S'$ is a vertex cover, if we take all elements in $E$ which are not part of $S'$ to create set $S''$ (line 9), there will be no two elements in $S''$ which have a conflict. In other words: all domains in $S''$ can be part of a feasible SDTP solution.

If $|S''| = |T|$ then every time-point has exactly one domain assigned to it, that is, $S''_i = e^c_{i},\ \forall\ i \in T$. If $|S''| < |T|$ the instance is infeasible (line 10). Line 11 continues to update graph $G_R$ with the information in $S''$ concerning the selected domain for each time-point:
\begin{align*}
 (\alpha,i,w_{\alpha i}) \in A_R \implies w_{\alpha i}=U(S''_i)\\
 (i,\alpha,w_{i\alpha}) \in A_R \implies w_{i\alpha}=-L(S''_i)
\end{align*}
\noindent where $U(S''_i)=u^c_i$ and $L(S''_i)=l^c_i$. The final solution $s$ is computed with standard \textsc{BellmanFord} since the SDTP has now been reduced to a feasible STP (line 12).

Note that in our implementation we do not explicitly create graphs $G_C$ and $G_B$. Instead, we directly create max-flow graph $G_F$. This graph is also modified by \textsc{SolveMaxFlow} to produce the residual graph. By taking this approach, we reduce both KA's execution time and the amount of memory it requires. Conceptually, however, it is easier to explain how KA works by documenting the creation of each graph in a step-by-step fashion.

\citet{art:kumar-estp} did not provide the asymptotic worst-case time complexity of KA and instead suggested that KA runs in polynomial time because each step can be performed in polynomial time. Therefore, for the purpose of completeness, we will now explicitly analyse the time complexity of KA. The three algorithmic components which dictate KA's complexity can be found on lines 1, 3 and 6. All other parts of the algorithm can be completed in time which is never slower than these three main components.

Line 1 takes time $O(|T|^2|C|)$ if computed with repeated \textsc{BellmanFord} and time $O(|T||C| + |T|^2\log |T|)$ if computed with Johnson's algorithm provided Dijkstra's algorithm \cite{book:cormen} is implemented with Fibonacci Heaps \cite{art:fib-heaps}. However, Fibonacci Heaps are often inefficient in practice due to pointer operations leading to poor cache locality and performance \cite{art:splib,art:heaps}. We therefore opted to use Sequence Heaps \cite{art:sequence-heaps} in our implementation, which increases the asymptotic worst-case time complexity to $O(|T||C|\log |T|)$ but improves the performance of the algorithm in practical settings. Line 3 has complexity $O(\omega^2)$ because we need to check every pair of intervals in $E$ and $|E| = O(\omega)$. Line 6 solves a max-flow problem. While there are many algorithms to solve max-flow \cite{art:max-flow}, we opted to use Dinic's Algorithm with complexity $O(\omega^\frac{5}{2})$ when applied to graphs from maximum bipartite matching. We have observed that max-flow is not the bottleneck in KA. Indeed, lines 1 and 3 are the most time-consuming steps (see Section \ref{sec:discussion} for a full discussion).

In the remainder of the paper, we will refer to the version of KA using repeated \textsc{BellmanFord} as KAB, and the one using Johnson's algorithm as KAJ. We will write KA when referring to the algorithm in a generic sense which covers both KAB and KAJ. The complexity of KAB is $O(|T|^2|C| + \omega^\frac{5}{2})$, while KAJ's is $O(|T||C|\log |T| + \omega^\frac{5}{2})$. Meanwhile, the space complexity of KA is $O(|T|^2 + \omega^2)$ due to distance matrix $\delta$ and graph $G_F$.

\subsection{Comin-Rizzi Algorithm}

\citet{art:cra} introduced asymptotically faster algorithms to solve both SDTPs and RDTPs, making their methods the current state of the art for both problems. For SDTPs, they introduced an algorithm which resembles Johnson's Algorithm for APSPs. Their method begins by performing a first phase using \textsc{BellmanFord} to detect negative cycles, while subsequent iterations use Dijkstra's Algorithm to correct computations over a graph that contains no negative cycles. However, no experimental study has been conducted using this method until now.

The \textit{Comin-Rizzi Algorithm} (CRA) for SDTPs is detailed in Algorithm \ref{alg:cra}. CRA begins by computing an initial earliest feasible solution $s^0$ considering $C_1$ constraints only. In our implementation, we partially consider $C_2$ constraints by using the global boundaries defined in Section \ref{sec:sdtp} within $G_D$. The computation of $s^0$ then either produces the earliest possible solution or proves that one cannot exist because (i) there is a negative cycle formed by $C_1$ constraints or (ii) it is not possible to assign a time $s_i$ to at least one time-point $i \in T$ while complying with the global bounds $[L(D_i),U(D_i)]$.

If $s^0 \neq \emptyset$ then CRA proceeds to its main loop. First, each time-point $i \in T$ where the current solution $s^0_i$ does not belong to one of the domains $D_i$ is added to list $F$ of assignments that require fixing. While there are elements in $F$, the following steps are repeated (lines 5-12). A time-point $i$ is removed from $F$ (line 6). The first time $i$ is removed from $F$, we compute entry $\delta_i$ of the distance matrix $\delta$ from $i$ to all other nodes in the underlying graph $G^{1\prime}_R$ containing only $C_1$ constraints (lines 7-9). In this graph, the weight $w_{ij}$ of each arc $(i,j,w_{ij}) \in A_R$ is modified to $w'_{ij}=w_{ij} + s^0_j - s^0_i$. \citet{art:cra} showed that $G^{1\prime}_R$ cannot contain negative cycles because it is always true that $w'_{ij} \geq 0$. Therefore, distances $\delta_i$ can be computed using \textsc{Dijkstra} instead of \textsc{BellmanFord}, which greatly improves the performance of CRA. Each entry $\delta_i$ is only computed once because $G^{1\prime}_R$ remains unchanged during CRA's execution.

\begin{algorithm}
  \caption{CRA}
  \label{alg:cra}
  \footnotesize
  \begin{algorithmic}[1] % The number tells where the line numbering should start
    \State $s^0 \gets \textsc{BellmanFord}(G_D, \alpha)$ \Comment{Solve STP using SDTP global boundaries}
    \IfThen{$s^0 = \emptyset$}{\textbf{return} $\emptyset$}
    \State $s \gets s^0$
    \State $F \gets \{i\ :\ (i,D_i) \in C_2 \land s_i \notin D_i\}$ \Comment{Set of all time-points $i \in T$ with assignment $s_i$ infeasible}
    \While{$F \neq \emptyset \textbf{ and } s \neq \emptyset \textbf{ and } s_i \leq U(D_i)\ \forall (i,D_i) \in C_2$}
    \State $i \gets \textsc{Pop}(F)$
    \If{$\delta_i \textbf{ not yet computed}$}
    \State $\delta_i \gets \textsc{Dijkstra}(G^{1\prime}_D, i)$ \Comment{Lazy computation of $\delta$}
    \EndIf
    \State $\textsc{UpdateAssignments}(s,s^0,i,\delta_i)$
    \State $F \gets \{i\ :\ (i,D_i) \in C_2 \land s_i \notin D_i\}$
    \EndWhile
    \State \Return $s$
  \end{algorithmic}
\end{algorithm}

For each $i$ taken from $F$ in line 6, we update the assignment to $s_i$ by means of procedure \textsc{UpdateAssignments} (line 10). First, the procedure performs the following operation
\begin{equation*}
    s_i \leftarrow \lambda(s_i,D_i)\ :\ (i,D_i) \in C_2
\end{equation*}
\noindent where $\lambda(s_i,D_i)$ is a function that either returns value $l^c_i$ belonging to the first domain in ascending order $[l^c_i,u^c_i] \in D_i$ for which $s_i < l^c_i$, or it returns $\perp$ if no such domain exists. Whenever $\lambda(s_i,D_i) = \perp$, CRA stops computations because this proves that the instance is infeasible. In this case, \textsc{UpdateAssignments} sets $s=\emptyset$. Alternatively, if $\lambda(s_i,D_i) \neq \perp$ then the new assignment $s_i$ can cause changes to other time-point assignments since $s_i$ has necessarily increased. To correctly propagate these changes, \citet{art:cra} introduced the following update rules
\begin{align*}
    \rho_{ij} &\leftarrow \delta_{ij} + (s_j - s^0_j) - (s_i - s^0_i),\ &\forall\ j \in P(G^1_D, i)\\
    s_j &\leftarrow s_j + \max(0, \lambda(s_i,D_i) - s_i - \rho_{ij}),\ &\forall\ j \in P(G^1_D, i)
\end{align*}
\noindent where $P(G^1_D,i)$ denotes the set of all nodes $j \in V$ which are reachable from $i$ in $G^1_R$. In other words: there is a path from $i$ to $j$ in $G^1_D$. These update rules can be applied in $O(1)$ time per $j \in P(G^1_D,i)$ or $O(|T|)$ time in total.

After fixing the assignment to $i$ and potentially other time-points, CRA constructs a new list $F$ (line 11). Once $F = \emptyset$, the assignment in $s$ is feasible and corresponds to the earliest feasible solution. This assignment is then returned in line 13. For infeasible instances, $s=\emptyset$ is returned instead.

The asymptotic worst-case time complexity of CRA is $O(|T||C| + |T|^2\log |T| + |T|\omega)$ when using Fibonacci Heaps for \textsc{Dijkstra}'s computations. The asymptotic complexity increases to $O(|T||C|\log |T| + |T|\omega)$ when using Sequence Heaps instead, however the empirical performance improves significantly \cite{art:sequence-heaps}. Regardless of the heap implementation, CRA's space complexity is $O(|T|^2)$ due to distance matrix $\delta$.

In their original description of CRA, \citet{art:cra} precomputed distance matrix $\delta$ before beginning the main loop in Algorithm \ref{alg:cra}. For our implementation, we describe the computation as a \textit{lazy computation} of entries in $\delta$ given that we only compute them when strictly necessary (lines 7-9). Although both approaches exhibit the same asymptotic worst-case time complexity, in practice the lazy computation performs significantly better since many unnecessary computations are avoided. Additionally, we have incorporated the creation of list $F$ at line 11 into procedure \textsc{UpdateAssignments}. Whenever the assignment $s_j$ to a time-point $j \in T$ is modified, we check whether $j$ should be added to or removed from $F$. This avoids reconstructing list $F$ every iteration of the main loop (lines 5-12), thus speeding up computations.

\subsection{Reduced Upper-Lower Tightening}

The \textit{Reduced Upper-Lower Tightening} (RULT) method is a speedup of ULT, specifically targeted towards SDTPs. One can easily derive RULT from ULT by exploiting the structure of SDTPs. Recall that in ULT, we must compute APSPs using \textsc{FloydWarshall} with complexity $O(|T|^3)$ because \citet{art:ult} assumed the input was a general DTP with possibly multiple disjunctions per constraint between two time-points $i,j \in T$. 

However, SDTPs feature a structure that only contains simple temporal constraints between time-points in $T$. It is therefore sufficient to compute single-source shortest paths twice: first to determine the earliest feasible assignment for each time-point and a second time to determine the latest feasible assignment for each time-point. This creates a single interval per time-point denoting a possibly  tighter global boundary concerning their assignments. Similar to ULT, we can use this global boundary to reduce $C_2$ disjunctions in every iteration, thereby reducing the number of disjunctions.

Algorithm \ref{alg:rult} outlines RULT. First, boundary set $B$ is initialized with the domains of each time-point (lines 1-2). In contrast to ULT, we only have to maintain boundaries per $i \in T$ rather than per constraint. The main loop (lines 3-11) runs for as long as there are changes to $B$ and the bounds remain feasible. In every iteration graph $G_D$ is changed with \textsc{UpdateGraph}, which replaces the weight of arcs connected to $\alpha$: 
\begin{align*}
 (\alpha,i,w_{\alpha i}) \in A_D \implies w_{\alpha i}=-L(B_i)\\
 (i,\alpha,w_{i\alpha}) \in A_D \implies w_{i\alpha}=U(B_i)
\end{align*}
\noindent The same procedure takes place for $G_R$ but outgoing arcs from $\alpha$ get the upper bound $U(B_i)$ while the incoming arcs get the lower bound $-L(B_i)$. During RULT's execution, values $L(B_i)$ are non-decreasing and $U(B_i)$ are non-increasing. Hence, updating the graphs tightens the global boundary $B_i$ of each time-point $i \in T$.

\begin{algorithm}
  \caption{RULT}
  \label{alg:rult}
  \footnotesize
  \begin{algorithmic}[1] % The number tells where the line numbering should start
    \State $B_i \gets \{[-\infty,+\infty]\},\ \forall\ i \in T$
    \State $B_i \gets D_i,\ \forall\ (i,D_i) \in C_2$
    \Do
    \State $\textsc{UpdateGraph}(G_D,B)$ \Comment{Update arc weights connected to $\alpha$}
    \State $\textsc{UpdateGraph}(G_R,B)$
    \State $p \gets \textsc{BellmanFord}(G_D,\alpha)$ \Comment{Earliest feasible assignment}
    \State $q \gets \textsc{BellmanFord}(G_R,\alpha)$ \Comment{Latest feasible assignment}
    \IfThen{$p = \emptyset$ \textbf{ or } $q = \emptyset$}{\textbf{return} $\emptyset$}
    \State $B'_i \gets \{[-p_i,q_i]\},\ \forall\ i \in T$
    \State $B \gets B \cap B'$ \Comment{Tightens boundaries}
    \DoWhile{$\textsc{Change}(B) \textbf{ and } L(B_{i}) \leq U(B_{i})\ \forall\ i \in T$}
    \State $s \gets \emptyset$
    \IfThen{$L(B_{i}) \leq U(B_{i})\ \forall\ i \in T$}{$s_i \gets L(B_{i}),\ \forall\ i \in T$}
    \State \Return $s$
  \end{algorithmic}
\end{algorithm}

Lines 6-7 compute the earliest feasible schedule $p$ and the latest feasible schedule $q$ over the updated graphs. If $p = \emptyset$ or $q = \emptyset$ then the instance is infeasible, because a negative cycle still exists even for the relaxed global boundaries of all time-points (line 8). Otherwise, line 9 constructs set $B'$ and line 10 computes the intersection of $B$ and $B'$. Operation $\cap$ is the same used in ULT and defined by \citet{art:ult}. Finally, lines 12-14 prepare solution $s$ to be returned. If boundaries in $B$ are feasible, line 13 assigns to every time-point its earliest feasible value. If the latest feasible solution is desired instead, we can assign $U(B_i)$ to $s_i$ in line 13.

The correctness of RULT follows directly from that of ULT \cite{art:ult} in combination with the fact that $C_1$ constraints are fixed and the only intervals that must be considered are those in $C_2$. The asymptotic worst-case time complexity of RULT is similar to ULT's. Accounting for the efficiency gain in shortest path computations, which are performed with \textsc{BellmanFord} instead of  \textsc{FloydWarshall}, RULT's complexity becomes: $O(|T|^2|C|K + |T|^2|K|^2)$. The space complexity of RULT is reduced to $O(|T|)$ given that we only have to allocate additional vectors of size $|T|$. Note that in our implementation, we do not explicitly maintain boundary set $B$.

\subsection{Bellman-Ford with Domain Check}

All of the algorithms described until now have employed \textsc{BellmanFord} at some point during their execution. This should not be surprising since \textsc{BellmanFord} can be implemented rather efficiently to detect negative cycles \cite{art:sp-fp}, which is a core task when solving STPs, SDTPs and RDTPs. It seems only natural then to consider a variant of the original algorithm to solve SDTPs. Let us therefore define \textit{Bellman-Ford with Domain Check} (BFDC), which incorporates small changes to \textsc{BellmanFord} in order to address gaps of infeasible values in the shortest path computations. Our method draws inspiration from previous research on temporal problems \cite{art:cra,art:stp-bf-inc,art:td-stp}.

Algorithm \ref{alg:bfdc} describes the full BFDC procedure, which primarily works over graph $G_D$. Lines 1-5 involve the initialization of auxiliary variables. This includes the distance array $\tau$, path length array $\pi$ which calculates the number of nodes in the shortest path from $\alpha$ up to $i \in V$, the domain index array $z$ which holds the current domain index $z_i$ for each time-point $i \in T$ and the first-in, first-out queue $Q$ used in \textsc{BellmanFord}. After initialization, the main loop begins (lines 6-20). In line 7, an element $i$ is removed from the queue and its domain is checked in line 8. Procedure \textsc{DomainCheck} is detailed in Algorithm \ref{alg:check}. If \textsc{DomainCheck} can prove the SDTP instance is infeasible, then it sets $\tau = \emptyset$. Otherwise the procedure updates assignments to $\tau$, $\pi$ and $z$ as necessary. The algorithm continues to line 9 where, if the instance has not been proven infeasible yet, all outgoing arcs from $i \in V$ are relaxed and the shortest paths propagated (here \textit{relax} refers to the nomenclature of \citet{book:cormen}).

\begin{algorithm}
  \caption{BFDC}
  \label{alg:bfdc}
  \footnotesize
  \begin{algorithmic}[1] % The number tells where the line numbering should start
    \State $\tau_i \gets +\infty,\ \forall\ i \in T\cup \{\alpha\}$
    \State $\pi_i \gets 0,\ \forall\ i \in T\cup \{\alpha\}$
    \State $z_i \gets 1,\ \forall\ i \in T\cup \{\alpha\}$
    \State $Q \gets \textsc{Push}(Q, \alpha)$
    \State $\tau_\alpha \gets 0$
    \While{$Q \neq \emptyset \textbf{ and } \tau \neq \emptyset$}
    \State $i \gets \textsc{Pop}(Q)$
    \State $\textsc{DomainCheck}(i,\tau,z,\pi)$ \Comment{Algorithm \ref{alg:check}}
    \If{$\tau \neq \emptyset$}
     \For{$(i,j,w_{ij}) \in A_D$}\Comment{Standard \textsc{Relax} phase in \textsc{BellmanFord}}
        \If{$\tau_j > \tau_i + w_{ij}$}
            \State $\tau_j \gets \tau_i + w_{ij}$
            \State $\pi_j \gets \pi_i + 1$
            \IfThen{$\pi_j \geq |T| \textbf{ or } j = \alpha$}{$\tau \gets \emptyset$} \Comment{Checks for negative cycle}
            \IfThen{$j \notin Q$}{$Q \gets \textsc{Push}(Q, j)$}
        \EndIf
        \IfThen{$\tau = \emptyset$}{\textbf{break}}
    \EndFor
    \EndIf
    \EndWhile
    \IfThenElse{$\tau \neq \emptyset$}{$s_i \gets -\tau_i,\ \forall\ i \in V$}{$s \gets \emptyset$}
    \State \Return $s$
  \end{algorithmic}
\end{algorithm}

For each outgoing arc from $i$ in $A_D$ (recall graph $G_D=(V,A_D)$), line 10 checks whether the current shortest path up to $j$ adjacent to $i$ should be updated and, if so, then the algorithm also updates $\pi_j$ and possibly queue $Q$. In line 14, if the path up to $j$ forms a cycle or the path leads back to $\alpha$, the instance is determined to be infeasible. If $j$ is not yet in queue $Q$, we add it in line 15 (duplicated elements are not allowed). When the instance has been proven infeasible, line 17 aborts the for-loop (lines 10-18).

The main loop runs for as long as there are elements in $Q$ and $\tau$ is not $\emptyset$. Once one of these conditions is false, Algorithm \ref{alg:bfdc} proceeds on to line 21. If the instance is feasible, assignment $s$ is created  using the values of the shortest paths stored in $\tau$, otherwise $\emptyset$ is returned.

Procedure \textsc{DomainCheck} (Algorithm \ref{alg:check}) represents the main difference between \textsc{BellmanFord} and BFDC. In line 1, it verifies whether the current time $s_i=-\tau_i$ assigned to $i$ belongs to its current domain indexed at $z_i$. If the assigned time does not exceed the domain's upper bound $u^{z_i}_i$ then \textsc{DomainCheck} simply terminates. Otherwise, the algorithm searches for the first domain in increasing order to which $s_i=-\tau_i$ belongs (lines 2-8). When a domain is found, lines 4-5 update the assignments for $\tau_i$ and $\pi_i$ accordingly. In case $s_i=-\tau_i$ exceeds all domains in $D_i$ then we have a proof that the instance is infeasible (line 9).

\begin{lemma}
    Algorithm \ref{alg:bfdc} is correct and returns either (i) the earliest feasible solution or (ii) proof that no solution exists. \label{lem:bfdc-correct}
\end{lemma}
\begin{proof}
    First, note that $\tau$ is always non-increasing in BFDC. This implies that the SDTP solution $s=-\tau$ is non-decreasing. In every \textsc{Relax} phase BFDC assigns the shortest path up to a subset of nodes in $G_D$ and therefore assigns the earliest feasible values to a subset of time-points. Whenever a \textsc{DomainCheck} phase must increase the assignment to $z_i$ because $-\tau_i > u^{z_i}_i$, it assigns the earliest feasible domain and either decreases $\tau_i$ or leaves $\tau_i$ unchanged (lines 4-5 in Algorithm \ref{alg:check}). Value $\tau_i$ is non-increasing and consequently decreasing the assignment of $z_i$ will never lead to a feasible solution. Therefore, $z$ is also non-decreasing in BFDC which implies domain assignment is a backtrack-free search.
    
    With these facts in mind, we can now show that there are two possibilities at the end of BFDC. If $\tau \neq \emptyset$ then $s=-\tau$ is the earliest feasible solution for the SDTP instance. This is true because $\tau$ contains the shortest paths in $G_D$ from $\alpha$ to every other node $i \in V$, with this achieved by using the minimum feasible assignment of domains $z$. When $\tau = \emptyset$ then we have either exhausted the assignment $z_i$ to a time-point $i \in T$ which implies that BFDC has run out of domains for $i$ ($z_i > |D_i|$), or there is a negative-cost cycle formed by $C_1$ constraints which has been detected during the \textsc{Relax} phase (line 14 in Algorithm \ref{alg:bfdc}).
\end{proof}

\begin{algorithm}
  \caption{\textsc{DomainCheck}}
  \label{alg:check}
  \footnotesize
  \begin{algorithmic}[1] % The number tells where the line numbering should start
    \Require{Time-point $i$, distance array $\tau$, domain index array $z$, path length array $\pi$}
    \If{$-\tau_i > u^{z_i}_i$}\Comment{Domain $[l^{z_i}_i,u^{z_i}_i] \in D_i$}
    \For{$z_i=z_i+1$ \textbf{ until } $|D_i|$}
    \If{$-\tau_i \leq u^{z_i}_i$}
        \State $\tau_i \gets \min\{\tau_i,-l^{z_i}_i\}$
        \IfThen{$\tau_i = -l^{z_i}_i$}{$\pi_i \gets 1$}
        \State \textbf{break}
    \EndIf
    \EndFor
    \IfThen{$-\tau_i > U(D_i)$}{$\tau \gets \emptyset$} \Comment{No domain can accomodate current $\tau_i$ assignment}
    \EndIf
  \end{algorithmic}
\end{algorithm}

It is possible to show that Lemma \ref{lem:bfdc-correct} also holds for the reversed case: producing the \textit{latest feasible solution}. For that, domains are sorted in descending order and computations occur over graph $G_R$ instead of $G_D$. This requires minor changes to how Algorithm \ref{alg:check} works to account for the reversed order of domains., with the general reasoning concerning how the algorithm operates remaining the same. The latest feasible solution $s$, if it exists, can be retrieved directly via $s=\tau$.  Let us now turn to the asymptotic worst-case time complexity of BFDC which is established via Lemma \ref{lem:bfdc-time}.

\begin{lemma}
    BFDC stops within a number of iterations proportional to $O(|T||C| + |T|\omega)$. \label{lem:bfdc-time}
\end{lemma}
\begin{proof}
    First, consider that the complexity of \textsc{BellmanFord} is $O(|T||C|)$ over the same graph $G_D$. The addition of \textsc{DomainCheck} does not change the size of queue $Q$ and therefore the overall number of iterations remains the same as standard \textsc{BellmanFord}. The change lies in the computational overhead of each iteration individually.
    
    There are at most $O(|V|)$ phases in \textsc{BellmanFord} with a first-in, first-out queue. In each phase, a node is extracted from $Q$ at most once \cite{book:networks}. In other words: the operations taking place in lines 7-11 of Algorithm \ref{alg:bfdc} are executed at most $|V|$ times per phase. These operations have a complexity equivalent to $O(\textsc{OutDeg}(i) + |D_i|)$, where \textsc{OutDeg}$(i)$ denotes the number of arcs in set $A_D$ which have $i$ as their source. Hence, each phase has complexity $O(\sum_{i \in V} \textsc{OutDeg}(i) + |D_i|)$ which is equivalent to $O(|C| + \omega)$. All together, we arrive at a complexity of $O(|V||C| + |V|\omega)$ which is equivalent to $O(|T||C| + |T|\omega)$ when solving SDTPs because $|V| = |T|$.
\end{proof}

The space complexity of BFDC is $O(|T|)$. The auxiliary arrays $\tau$, $s$, $\pi$, $z$ and queue $Q$ used in BFDC all require additional space proportional to $|T|$.

\subsection{Asymptotic worst-case complexities}

Let us now assess the theoretical complexities of all the algorithms and draw some initial conclusions concerning what one should expect from empirical results. Table \ref{tab:complexities} provides both the asymptotic worst-case time complexity and space complexity for each algorithm according to our implementation. Given that the ILP, CP and SCP models are often solved by means of general-purpose black-box solvers, we opted not to include their theoretical complexities in our analysis.

   \begin{table}[!htb]
       \centering
       \caption{Asymptotic worst-case complexities for each algorithm.}
       \label{tab:complexities}
       \begin{tabular}{lrr}
       \hline
            Algorithm & Time complexity & Space complexity\\
            \hline
            ULT & $O(|T|^3|C|K + |C|^2K^2)$ & $O(|T|^2)$\\
            KAB & $O(|T|^2|C| + \omega^\frac{5}{2})$ & $O(|T|^2 + \omega^2)$\\
            KAJ & $O(|T||C|\log |T| + \omega^\frac{5}{2})$ & $O(|T|^2 + \omega^2)$\\
            CRA & $O(|T||C|\log |T| + |T|\omega)$ & $O(|T|^2)$\\
            RULT & $O(|T|^2|C|K + |T|^2K^2)$ & $O(|T|)$\\
            BFDC & $O(|T||C| + |T|\omega)$ & $O(|T|)$\\
            \hline
       \end{tabular}
   \end{table}

    In terms of worst-case time complexity, BFDC clearly outperforms all other methods. CRA is the second fastest method. Meanwhile, it is difficult to rank KA and RULT because they have different terms which can dominate one another. Note that $\omega = O(|T|K)$ and therefore whenever $\omega \leq (|T|K)^{\frac{4}{5}}$ the second term (max-flow) in KA's complexity is never slower than RULT's second term. In this case, we can limit our comparison to the first term referring to shortest paths. Clearly, both KAB and KAJ are asymptotically faster than RULT in this regard. However, when $\omega \approx |T|K$ the time complexity of KA is lower than RULT's due to the max-flow phase. As previously mentioned, we can also see that the use of Johnson's Algorithm in KA reduces its time complexity, bringing KAJ closer to CRA. Finally, ULT is the slowest algorithm in Table \ref{tab:complexities}, mainly due to its heavy utilization of \textsc{FloydWarshall}.  
    
    The time complexities documented in Table \ref{tab:complexities} are indicative of the challenges faced when solving SDTPs. Despite their close ties to shortest path problems, the presence of negative cycles and disjunctive domains requires more complex and refined techniques. In particular, we wish to call attention to the increased space complexity in most of the established techniques in the literature. Only RULT and BFDC are able to solve SDTPs using linear space. Although this may appear unimportant given the availability of computational resources, a quadratic memory overhead can quickly become prohibitive in practice. This is often problematic given that SDTPs appear as subproblems of other more complex problems which require their own share of memory. We will discuss the impact of memory usage later in Section \ref{sec:discussion}. 
    
    A final remark concerns the relation between SDTPs and time-dependent STPs established in Section \ref{sec:sdtp}. \citet{art:td-stp} showed how the time-dependent STP can be solved in time $O(|T||C|)$. However, despite the relation between the two problems, Table \ref{tab:complexities} shows that solving SDTPs requires asymptotically more time than time-dependent STPs in the worst case. This is primarily due to the discontinuity of time-point domains which renders certain assignments in SDTP solutions infeasible, thereby requiring additional procedures to correct the assignments and (re)check feasibility. When at most one domain exists per time-point ($K \leq 1$), this correction is not necessary.

\section{Experiments} \label{sec:exps}

Experiments were carried out on a computer running Ubuntu 20.04 LTS equipped with two Intel Xeon E5-2660v3 processors at 2.60GHz, with a total of 160 GB RAM, 5 MB of L2 cache and 50 MB of L3 cache. Intel's Hyper-threading technology has been disabled at all times to avoid negatively influencing the experiments. All of the algorithms were implemented using \texttt{C++} and compiled with GNU GCC 9.3 using optimization flag \texttt{-O3}. The ILP, CP and SCP models were implemented using the \texttt{C++} API of CPLEX 12.9. Methods were only allowed one thread during execution.

Our experiments primarily focus on measuring observed computation times. In order to obtain accurate time measurements, we employ \texttt{C++}'s \texttt{std::steady\_clock} to measure CPU time. To ensure as much fairness as possible when comparing methods that differ significantly with respect to the input representation, we decided to document only the computation time for solving an instance. This means our results do not include information concerning the time needed for input, output or preprocessing that is performed by some algorithms to transform data into a more suitable format. Similarly, the time to build ILP, CP and SCP models is not included in their results.

While we understand that evaluating methods with respect to computation times is not always ideal \cite{art:exp1,art:exp2}, it is difficult to obtain a single evaluation metric for algorithms that differ so much in terms of their basic operations and components. Additionally, \citet{art:shapiro} argued that for tractable problems, the running time of algorithms is often a reasonable metric. 

Since we are proposing the first experimental study to evaluate algorithms for solving SDTPs, we introduce four datasets to assess the performance of the various methods and their implementations. The four datasets differ in terms of their problem structure and particularly with respect to the underlying distance graph. Instances are subdivided into \textit{shortest path instances}, \textit{negative cycle instances}, \textit{vehicle routing instances} and \textit{very large instances}.  All of them include only integer values. It is therefore possible to accommodate all methods, including CP and SCP, without any changes. We will begin by first detailing the procedure by which we generated each instance set before presenting the computational results obtained from our experiments.

\subsection{Shortest path instances} \label{sec:sp-instances}

Instances are created with graph generators for Shortest Path (SP) problems. A graph $G=(V,A)$ is transformed into an SDTN $N=(T,C)$ by setting $T = V$, $C_1 = A$ and deriving $C_2$ constraints for the time-point domains from the shortest paths in $G$. Let us define the following parameters for an SDTP instance: number of time-points $|T|$, number of Type 1 constraints $|C_1|$, number of elements with more than one domain $|T_D|$, and number of domains $K > 1$ per $i \in T_D$ such that $|D_i| = K$. There are four SP groups which differ in terms of how either graph $G$ or $C_2$ constraints are created. The generation process for each group is summarized below (for more details see Appendix \ref{ap:instances}).

\begin{enumerate}
    \item \textsc{Rand}: generates graph $G$ using \textsc{Sprand} introduced in the SPLib \cite{art:splib}. Nodes and arcs are all created randomly. Constraints $C_2$ are generated based on the shortest path from a dummy node to every $i \in V$. 
    \item \textsc{Grid}: generates graph $G$ using \textsc{Spgrid}, also introduced in the SPLib \cite{art:splib}. Nodes are generated in a grid format with $X$ layers and $Y$ nodes per layer. Arcs connect nodes within the same layer and to those in subsequent layers. $C_2$ constraints are generated in the same way as for \textsc{Rand}. 
    \item \textsc{Seq}: generates graph $G$ using the tailored generator \textsc{Spseq}. Nodes are generated at random similarly to \textsc{Sprand}. A path connecting all nodes with $|V|-1$ arcs is created where the weight of all arcs is $w_{ij}=1$. Afterwards, the remaining $|A|-|V|+1$ arcs are created at random with greater weights. This creates a known shortest path which may be difficult for some methods to find. $C_2$ constraints are generated in the same way as for \textsc{Rand} and \textsc{Grid}.
    \item \textsc{Late}: generates graph $G$ using either \textsc{Sprand} or \textsc{Spseq}. $C_2$ constraints are created so that at least 60\% of the earliest feasible solutions $s_i$ belong to the last domain of the respective time-point. 
\end{enumerate}

For each of these four datasets, we also create four subsets to assess which key instance characteristics have the biggest impact on algorithmic performance. For each of these subsets we fix three of the parameters of an SDTP, and then vary the fourth. 

\subsubsection{\textsc{Nodes} dataset}

The number of time-points $|T|$ varies in the range \{100,200,\dots,12800,25600\} for dataset \textsc{Nodes}. Other parameters are fixed to $|C_1| = 6\cdot|T|$, $|T_D|=0.8\cdot|T|$ and $K=10$. Five instances were generated for each combination of dataset (1)-(4) and number of time-points $|T|$ (henceforth denoted a configuration): three feasible and two infeasible instances. For example, five instances have been generated for configuration (\textsc{Rand}, \textsc{Nodes}, $|T|=100$). 

Figure \ref{fig:sp-nodes} provides the results for the \textsc{Nodes} subset. Each graph reports the average computation times for each method according to the number of time-points for each dataset. The values reported are the average from 20 runs, so as to mitigate the impact of any outliers due to the short computation times needed to solve SDTPs \cite{art:shapiro}. Additionally, methods are given a time limit of two seconds. If a method timed out for all instances of a given size, we omit these results for clarity. This explains the incomplete curves present in some of the graphs. However, if for some instance sizes a method could solve at least one instance (out of five), we report the averages including potential timeouts. KAJ rarely outperformed KAB. Based on these results, we decided to only show the results for KAB. In the experiments, we will comment on specific differences between the two methods whenever necessary.

\begin{figure}[!htb]
    \centering
    \includegraphics[width=\linewidth]{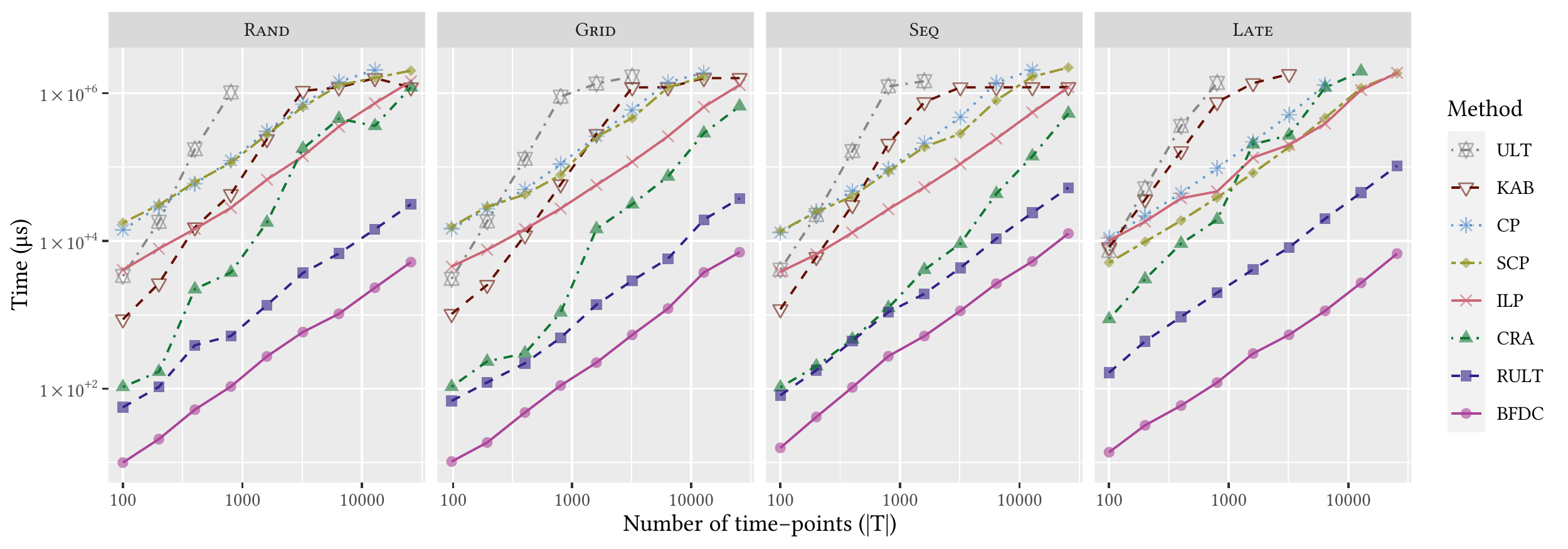}
    \caption{Average computation times in microseconds ($\mu$s) for the \textsc{Nodes} subset. Both $x$- and $y$-axis are reported in log scale.}
    \label{fig:sp-nodes}
\end{figure}

The general performance of the algorithms in graphs of Figure \ref{fig:sp-nodes} is somewhat consistent. We can clearly see a cluster of curves towards the top of the graphs which include ULT, CP, SCP, ILP and KAB. It is also easy to distinguish a second cluster formed by RULT and BFDC at the bottom of the graphs. Meanwhile, CRA lies in-between these two clusters, typically starting at the bottom for small instances and trending towards the top for the largest ones. These differences are not surprising given that most methods in the top cluster are more general than those in the bottom cluster (including CRA). While these differences are to be expected, the question as to whether they hold in different scenarios must still be answered.

ULT demonstrates the fastest growth in the graphs and can rarely solve instances containing $|T| > 1000$. This behavior is easily explained by the use of \textsc{FloydWarshall} in every iteration, which contributes to a $\Theta(|T|^3)$ time complexity. CP and SCP are relatively consistent in execution time, with SCP able to solve slightly larger instances in \textsc{Rand}, \textsc{Seq} and \textsc{Late}. For \textsc{Late} instances, SCP outperformed CP in all scenarios. This showcases how CPLEX as a CP solver can benefit from SCP's model structure. Meanwhile, ILP is the quickest method in the topmost cluster of methods for \textsc{Rand}, \textsc{Grid} and \textsc{SEQ}, while it performs similarly to SCP for \textsc{Late} instances.

One can observe that KAB only outperforms other methods in the topmost cluster for small instances ($|T| \leq 800$). The reason is that computing APSPs requires a significant amount of time and quickly becomes prohibitive for larger instances. For \textsc{Late} instances, KAB was unable to solve those where $|T| > 3200$. This is because \textsc{Late} instances tend to have larger interval sets $E$, which directly impact the creation of the conflict graph (line 3 in Algorithm \ref{alg:kra}) thereby limiting KA's execution time.

CP, SCP, ILP, ULT, and KAB always take longer than one millisecond to compute results. Meanwhile, CRA begins below or at this threshold in all cases and grows quickly, often reaching the one-second threshold. Nevertheless, we can see that CRA typically performs better than ILP. Even when CRA is slower, the differences are not significant. This is true except for the \textsc{Late} instances, where CRA timed out for all five instances with $|T| = 25600$. In these cases, the combination of many time-points and late feasible schedules leads CRA to compute more entries of the distance matrix using \textsc{Dijkstra}, causing major computational overhead for the method (line 8 in Algorithm~\ref{alg:cra}). Both RULT and BFDC always remain below the 100 milliseconds threshold and are faster than CRA. This is despite the fact that RULT has a theoretical worst-case time complexity slower than CRA. We did not observe any timeout for either method in the bottom cluster, which contributes to their lower curves in Figure \ref{fig:sp-nodes}. The relative position of both algorithms is also very consistent, with RULT only slightly slower than BFDC.

\subsubsection{\textsc{Density} dataset}
The second subset is \textsc{Density}, in which we vary the number of constraints $C_1$ in the range  \{20,25,\dots,85,90\}\%, given as a percentage of the maximum number of constraints (maximum number of arcs in the base graph). Other parameters are fixed as follows: $|T|=1008$, $|T_D|=0.8\cdot|T|$ and $K=10$. Similar to \textsc{Nodes}, five instances are generated per configuration. Average computation times according to density growth are shown in Figure \ref{fig:sp-density}.

\begin{figure}[!htb]
    \centering
    \includegraphics[width=\linewidth]{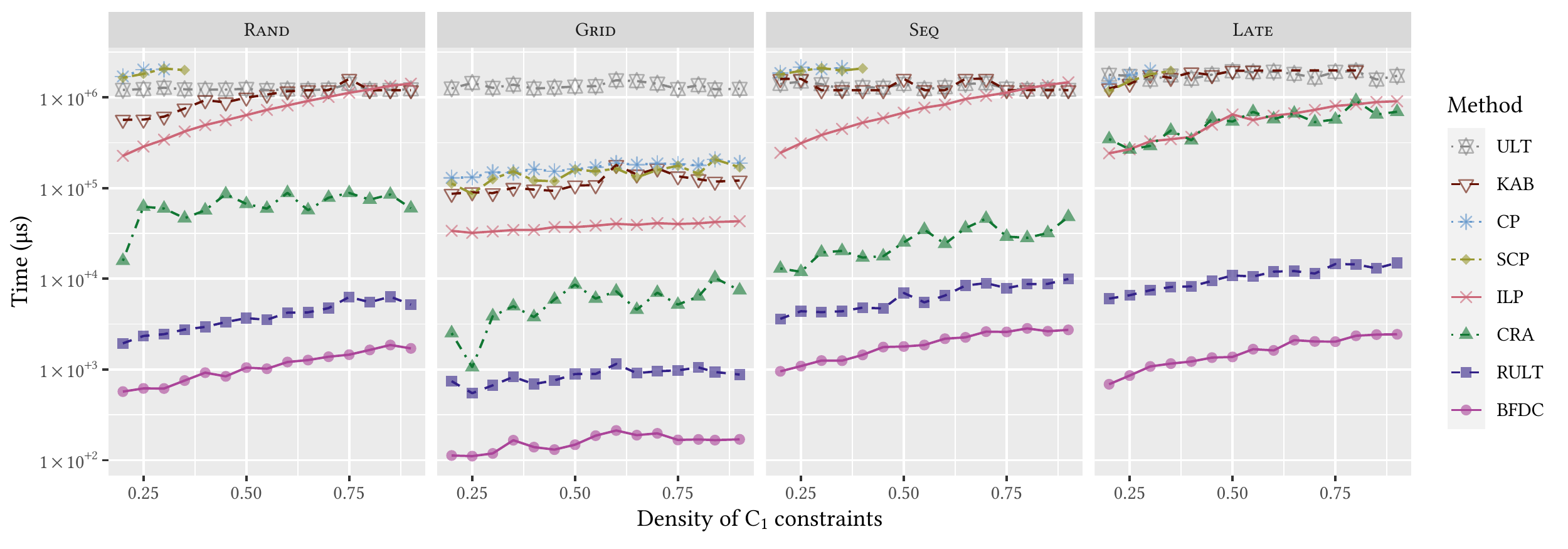}
    \caption{Average computation times for the \textsc{Density} subset. The $y$-axis is reported in log scale.}
    \label{fig:sp-density}
\end{figure}

On the one hand, the performance of ULT barely changes with respect to varying densities due to \textsc{FloydWarshall}'s phase which maintains ULT among the slowest methods. On the other hand, ULT was always able to solve at least one instance per configuration which is not true for all methods. Except for the \textsc{Grid} instances, CP and SCP experience difficulties solving problems with a density greater than $40\%$. KAB far outperforms ULT for \textsc{Grid} instances, but generally performs similarly to ULT in all other cases. KAB is also never slower than CP or SCP. In terms of the topmost cluster, ILP was consistently the best performing method.

Here we notice that CRA demonstrates far better performance than ULT, CP, SCP, ILP and KAB for the \textsc{Rand}, \textsc{Grid} and \textsc{Seq} instances. It remains a sort of middling algorithm, but shows little variation in performance resulting from the network's density. However, for the \textsc{Late} instances, CRA's performance compares to that of the ILP.
Similar to the \textsc{Nodes} dataset, this is explained by the overhead incurred by \textsc{Dijkstra} computations. For the \textsc{Density} instances, however, the network is more connected, leading to more time-points being affected by changes made to others. This in turn also requires more assignment updates.

Both RULT and BFDC also show little variation with respect to the network's density. They remain the fastest algorithms, with no timeouts observed. 

\subsubsection{\textsc{NumDisj} dataset} 

In the third subset, we vary the number of domains $K$ per time-point $i \in T_D$ in the range \{5,10,20,\dots,90,100\}. Other parameters are fixed as follows: $|T|=2000$, $|C_1| = 12000$ and $|T_D|=1600$. Figure \ref{fig:sp-numd} shows the average computation times according to parameter $K$.

\begin{figure}[!htb]
    \centering
    \includegraphics[width=\linewidth]{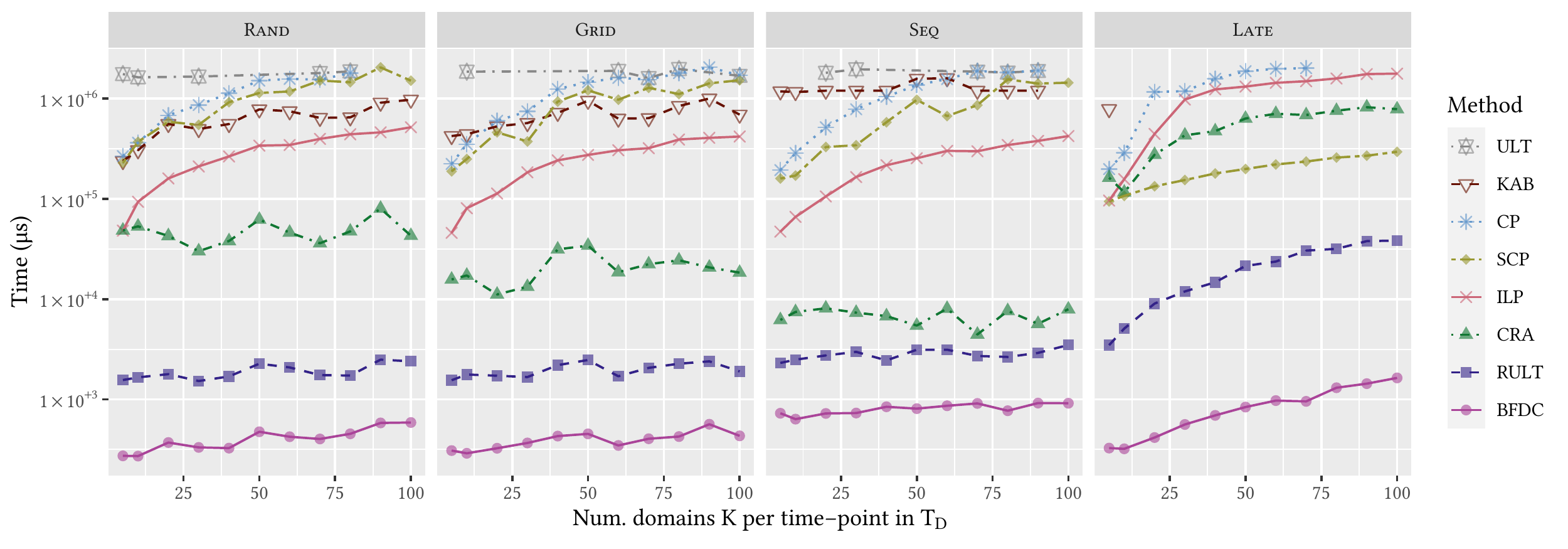}
    \caption{Average computation times for the \textsc{NumDisj} subset. The $y$-axis is reported in log scale.}
    \label{fig:sp-numd}
\end{figure}

In these experiments, we can still differentiate the three clusters of methods from before, but their individual behaviors are now distinct. ULT suffers far more timeouts and often cannot solve a single instance in any configuration. However, this appears unrelated to parameter $K$ and more due to some other specific instance characteristic that was not captured in these experiments. CP and SCP can solve most instance sizes. Additionally,  SCP is faster and can solve more instances as the value of $K$ increases compared to CP. This difference is pronounced for dataset \textsc{Late}, where SCP not only outperforms CP but also CRA and ILP in all cases. Despite its slow performance for the \textsc{Late} instances, ILP outperforms ULT, CP, SCP and KAB across all other configurations.

KAB experiences the same difficulties solving \textsc{Late} instances, where only the smallest ones with $K=5$ were solved. This is unsurprising since $\omega$ is directly related to $K$ (recall Section \ref{sec:sdtp}). KAJ performed slightly better than KAB and was able to solve \textsc{Seq} instances with $K \leq 90$. CRA again outperforms all methods in the topmost cluster except for the \textsc{Late} instances, where SCP is faster. For \textsc{Seq}, we notice the power of these \textsc{Dijkstra} computations because they help CRA to easily find the hidden shortest path used during the instance's construction. This then leads to much shorter executions. RULT and BFDC remain the fastest methods. However, RULT is clearly impacted to a far greater extent by the growth in the number of time-point domains compared to BFDC for the \textsc{Late} instances. This observation is aligned with their asymptotic worst-case time complexities.

The graphs in Figure \ref{fig:sp-numd} suggest that the methods solved using CPLEX (CP, SCP and ILP) are those most impacted by increases in $K$. This may be due to the number of constraints created when more domains exist and an increase in the cardinality of sets $\Phi_i$ in the SCP model.

\subsubsection{\textsc{VarDisj} dataset} The fourth and final subset for SP instances is \textsc{VarDisj}. This subset varies the size of $T_D$ in the range \{10,20,\dots,90,100\}\%, given as a percentage of the total number of time-points $|T|$ that have multiple domains. Other parameters are fixed as follows: $|T|=2000$, $|C_1|=6\cdot|T|$ and $K=10$. Figure \ref{fig:sp-vard} provides computation runtime results according to the size of set $T_D$.

\begin{figure}[!htb]
    \centering
    \includegraphics[width=\linewidth]{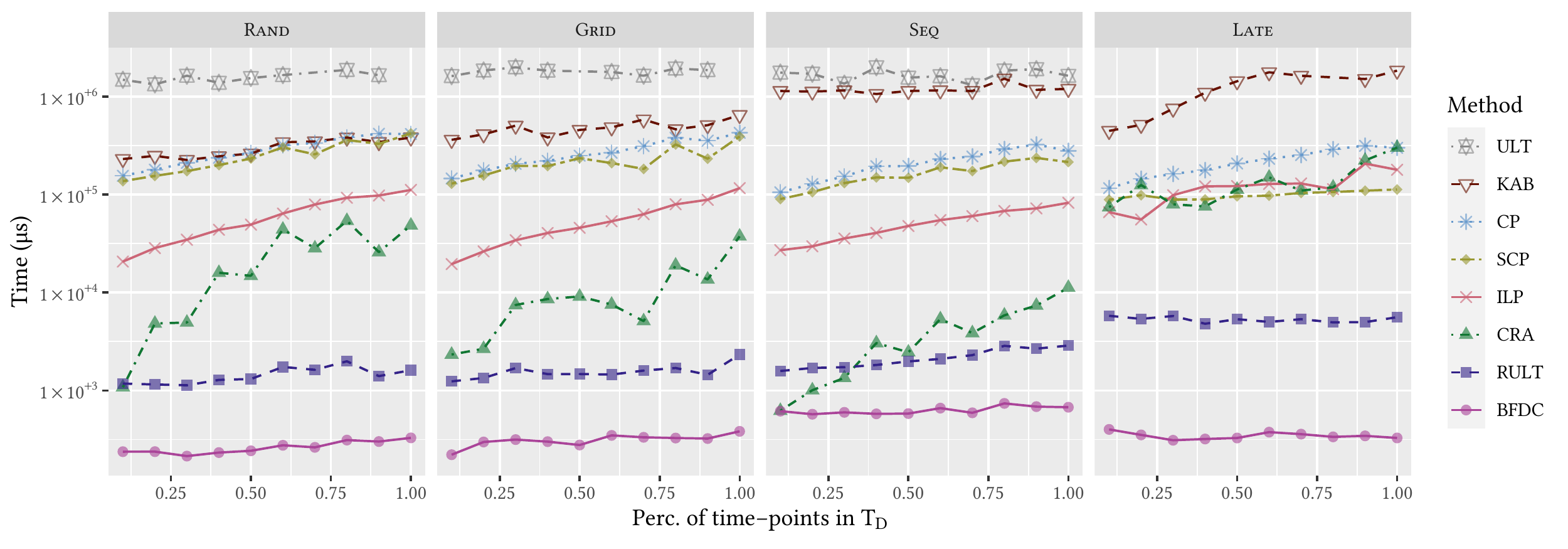}
    \caption{Average computation times for the \textsc{VarDisj} subset. The $y$-axis is reported in log scale.}
    \label{fig:sp-vard}
\end{figure}

Figure \ref{fig:sp-vard} shows how the results for this subset of instances differ significantly from the previous subsets. First, for instances \textsc{Rand}, \textsc{Grid} and \textsc{Seq} the methods are now far more distinctly dispersed across different parts of the graphs. ULT is once again the slowest method, highlighting its difficulty in computing APSPs. KAB competes with CP and SCP in dataset \textsc{Rand}, but is clearly slower than these methods for \textsc{Grid} and \textsc{Seq}. Similar to \textsc{NumDisj}, KAJ demonstrated slightly better performance for \textsc{Seq} instances than KAB, yet not enough to outperform CP or SCP. Indeed, these two methods always outperform KAB and ULT. Meanwhile, ILP is quicker than the previous four methods, maintaining its somewhat consistent behavior as the best general-purpose algorithm for solving SDTPs. CRA is clearly the method that suffers the most from an increasing number of time-points that have multiple domains. Nevertheless, for the first three datasets, CRA is always faster than ILP and for the \textsc{Seq} instances even outperforms RULT for small ones. Finally, RULT and BFDC remain the fastest methods and as the size of set $T_D$ increases there is little noticeable impact on their performance.

When we consider the \textsc{Late} instances, which are arguably the most difficult to solve, the situation changes. ULT cannot solve a single instance in this set. KAB appears as the slowest method. Meanwhile, CP, SCP, ILP and CRA are all clustered below KAB. CP is clearly the slowest of the four methods. ILP and CRA exhibit some variations, but overall their growth trend is far less pronounced than for the other three sets (\textsc{Rand}, \textsc{Grid} and \textsc{Seq}). SCP is very consistent across all experiments and for large $T_D$ sets outperforms the other three algorithms, albeit not by a very large margin. RULT is a whole order of magnitude slower than BFDC for almost all datasets except for \textsc{Seq}. Nevertheless, neither RULT nor BFDC exhibit any significant variations in performance.

\subsection{Negative cycle instances}

The \textsc{Negcycle} instances use the filter of same name proposed by \citet{art:sp-fp}. This filter is applied to instances from Section \ref{sec:sp-instances} by introducing a negative cycle into the underlying base graph. Only instances that are feasible before applying the filter are considered so that they become infeasible precisely due to the negative cycle. Similar to \citet{art:sp-fp}, we consider four classes of negative cycles: (\textsc{Nc02}) one cycle with three arcs; (\textsc{Nc03}) $\lfloor \sqrt{|T|} \rfloor$ cycles with three arcs each; (\textsc{Nc04}) $\lfloor \sqrt[3]{|T|} \rfloor$ cycles with $\lfloor \sqrt{|T|} \rfloor$ arcs each; and (\textsc{Nc05}) one Hamiltonian cycle. 

For each one of these four classes, we vary $|T|$ in the range \{100,200,\dots,12800,25600\} while fixing parameters $|C_1|=6\cdot|T|$, $|T_d|=0.8\cdot|T|$, $K=10$. For each configuration, we generate three instances. Figure \ref{fig:nc} provides the computation results for each of the negative cycle classes, with the average time to prove infeasibility reported for each method over 20 runs.

\begin{figure}[!htb]
    \centering
    \includegraphics[width=\linewidth]{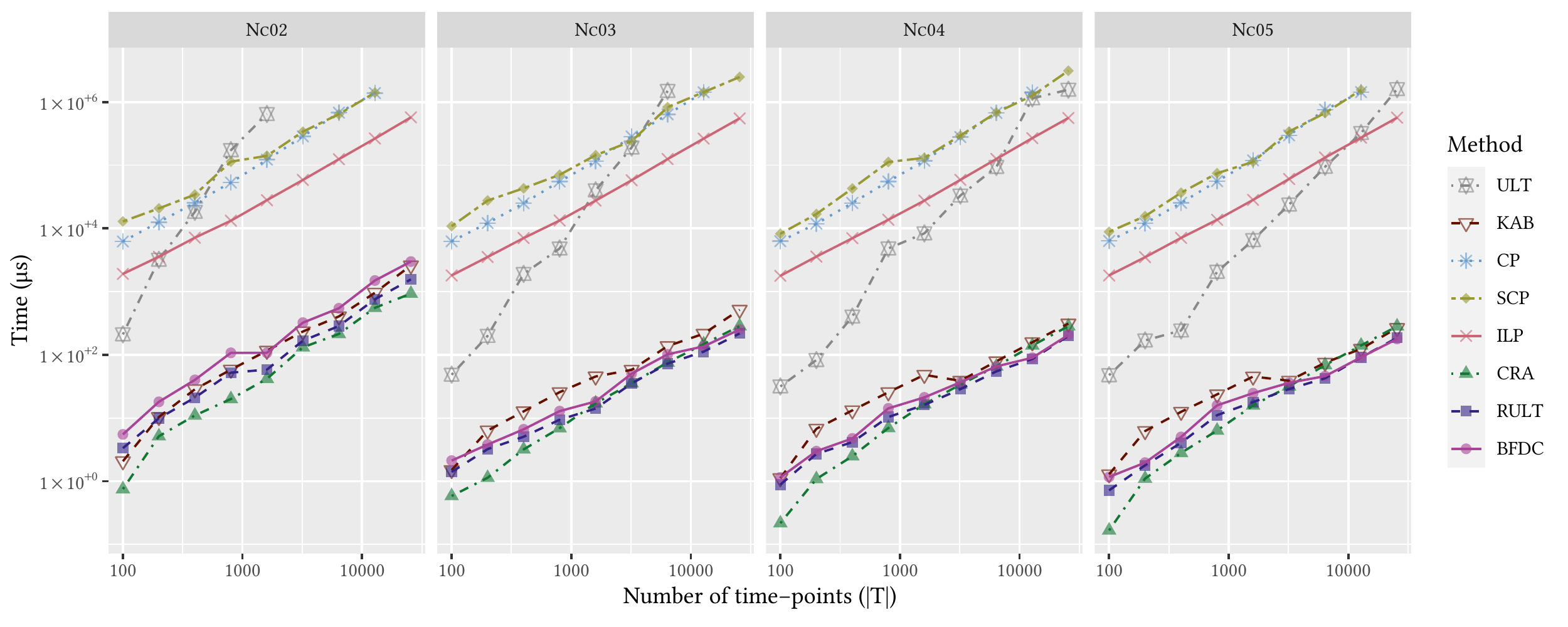}
    \caption{Average computation times for the \textsc{Negcycle} datasets. Both the $x$- and $y$-axis are reported in log scale.}
    \label{fig:nc}
\end{figure}

An initial observation is that ULT experiences greater difficulty solving instances with fewer and smaller cycles (\textsc{Nc02} and \textsc{Nc03}). CP was unable to solve the largest instances in any of the cases, while SCP was able to prove infeasibility of the large instances when more than one cycle existed (\textsc{Nc03} and \textsc{Nc04}). Meanwhile, ILP demonstrated a very consistent performance despite the cycles.

For the first time, we notice KAB among the fastest algorithms at the bottom of the graphs. This is not surprising since for \textsc{Negcycle} instances KAB is able to detect the cycle when computing distance matrix $\delta$ (line 1 in Algorithm \ref{alg:kra}). Also for the first time, BFDC demonstrates consistently the lowest execution time for \textsc{Nc02} compared to CRA, RULT and KAB. For the same dataset, CRA is consistently the fastest method because it can also detect negative cycles in its first phase with standard \textsc{BellmanFord} (line 1 in Algorithm \ref{alg:cra}). This showcases the overhead incurred by \textsc{DomainCheck} in BFDC (line 8 in Algorithm \ref{alg:bfdc}).

However, in datasets \textsc{Nc03}-\textsc{Nc05} the methods which feature in the cluster at the bottom are harder to differentiate. KAB typically appears to be the slowest, CRA the fastest, with BFDC and RULT lying somewhere in the middle. Nevertheless, the differences between KAB, RULT, BFDC and CRA for \textsc{Negcycle} instances are negligible for most purposes.

\subsection{Vehicle routing instances}

We extract vehicle routing instances (\textsc{Vrp}) from solutions to Vehicle Routing Problems with Multiple Synchronization (VRPMS) constraints \cite{art:vrpms}. These VRPMS instances contain multiple routes for which departure times and service times must be assigned while complying with synchronization constraints between routes, in addition to maximum route duration constraints. We refer interested readers to Appendix \ref{ap:instances} for more information concerning how exactly these instances have been generated.

\textsc{Vrp} instances primarily differ in terms of their number of time-points, which ranges from 10 to 1300. For each instance size, we again create five instances: three feasible and two infeasible. Figure \ref{fig:vrp} presents the computation times per method according to the number of time-points in the instance, with the values reporting the average over 20 runs. Results are grouped according to instance feasibility given that some differences in performance can be observed depending on whether an instance is feasible.

\begin{figure}[!htb]
    \centering
    \includegraphics[width=\linewidth]{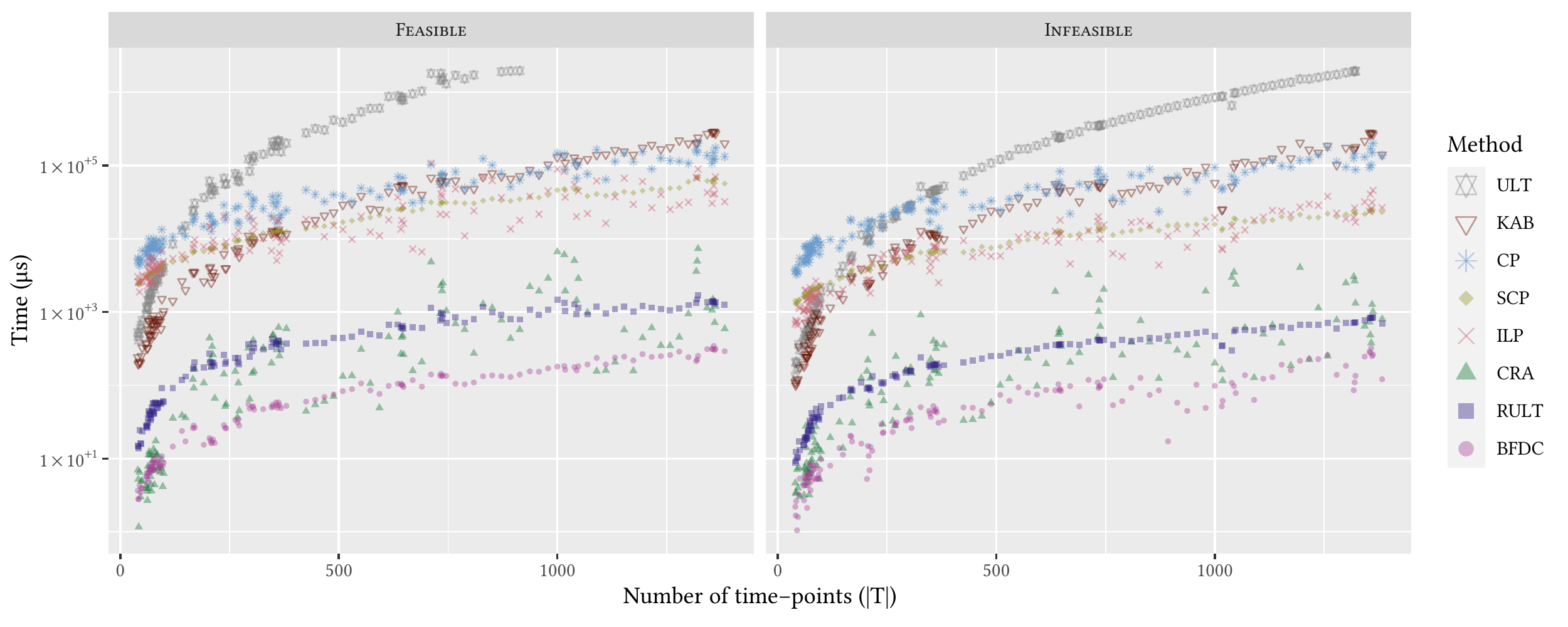}
    \caption{Average computation times for the \textsc{Vrp} dataset. The $y$-axis is reported in log scale.}
    \label{fig:vrp}
\end{figure}

One can quickly notice that ULT can easily prove infeasibility of instances, but it experiences difficulty producing solutions for feasible instances. Indeed, for feasible instances ULT cannot solve those where $|T| > 1000$. CP and SCP exhibit different performances, with SCP faster on average, particularly for infeasible instances. KAB consistently performs similarly to CP. The performance of ILP varies significantly for feasible instances and somewhat less significantly for the infeasible cases. It is difficult to conclude whether ILP is faster than SCP overall, although we can easily see that both methods are faster than ULT, CP and KAB. 

CRA, RULT and BFDC are once again clustered towards the bottom of the graph, signifying that they are the fastest methods. RULT is the slowest among the three, although CRA does vary a lot and is slower for certain cases. BFDC can be concluded to be the fastest method although, yet again, CRA does outperform BFDC in certain cases. Overall though, one can summarize the order from slowest to fastest as follows: RULT, CRA and BFDC.

Finally, note that \textsc{Vrp} instances all have an underlying distance graph which is very sparse. Vertices have at most two outgoing arcs and two incoming arcs, with the exception of those connected to the beginning of the time horizon $\alpha$. Additionally, the graphs are almost acyclical and contain a very limited number of arcs that create cycles. While we do not exploit this structure when solving the \textsc{Vrp} dataset, this would certainly be an interesting avenue for future research.

\subsection{Very large instances}

Although the three previous datasets have diverse characteristics, they all fail to capture scenarios where the number of time-points is very large. These types of instances are important when it comes to truly verifying the scalability of methods which may have advantages for small-scale problems yet suffer when instance size grows significantly. To verify the extent to which our previous analysis holds for such instances, we generate five very large (\textsc{Vl}) problems, all of which are feasible. Table \ref{tab:vl} presents the characteristics of these instances in terms of their \textit{Base} generation method, number of time-points, number of $C_1$ constraints, maximum number $K$ of domains per time-point and total number of domains $\omega$. Base \textsc{Tsp} means that graph $G$ was extracted from the Traveling Salesman Problem Library (TSPLib) \cite{art:tsplib} instance \texttt{pla85900} that contains 85900 nodes. For instance \textsc{Vl}-1, a random subset of nodes is selected from \texttt{pla85900}. Meanwhile, instances \textsc{Vl}-3, \textsc{Vl}-4 and \textsc{Vl}-5 are generated with the SP procedures outlined in Section \ref{sec:sp-instances}.

\begin{table}[!ht]
\caption{Very large instances.}\label{tab:vl}
\centering
\begin{tabular}{llrrrr}
\hline
Instance & Base & $|T|$ & $|C_1|$ & $K$ & $\omega$\\
\hline
\textsc{Vl-1} & \textsc{Tsp} & 50 000 & 500 000 & 20 & 905 000\\
\textsc{Vl-2} & \textsc{Tsp} & 85 900 & 859 000 & 60 & 4 647 190\\
\textsc{Vl-3} & \textsc{Seq} & 200 000 & 2 000 000 & 100 & 16 040 000\\
\textsc{Vl-4} & \textsc{Late} & 400 000 & 4 000 000 & 180 & 57 680 000\\
\textsc{Vl-5} & \textsc{Rand} & 1 000 000 & 10 000 000 & 500 & 400 200 000\\
\hline
\end{tabular}
\end{table}

Figure \ref{fig:vl} provides a graph documenting the average execution times over 10 runs for a subset of the methods with the \textsc{Vl} instances. We opted to test only the best performing methods: ILP, CRA, RULT and BFDC. For each run, the methods were given a one hour time limit.

Even among the four best performing methods, it is obvious that CRA and ILP are limited with respect to the instance sizes they can solve. Indeed, they were unable to solve instances beyond \textsc{Vl-3} within one hour of execution. While it is difficult to determine the precise reason for the behavior of the ILP solver, the reason for CRA lies in how large instances require more \textsc{Dijkstra} computations (in particular for \textsc{Vl-4}). The computation of APSPs is not only slow but also leads to far worse memory locality. This impacts cache usage which reduces the performance of CRA compared to RULT, even though the latter has a theoretically slower worst-case time complexity. Figure \ref{fig:vl} also showcases the fact that \textsc{Late} instances are much more difficult than other instances. This is true even when a \textsc{Late} instance has less than half the number of time-points of a \textsc{Rand} instance.  

\begin{figure}[!htb]
    \centering
    \includegraphics[width=0.9\linewidth]{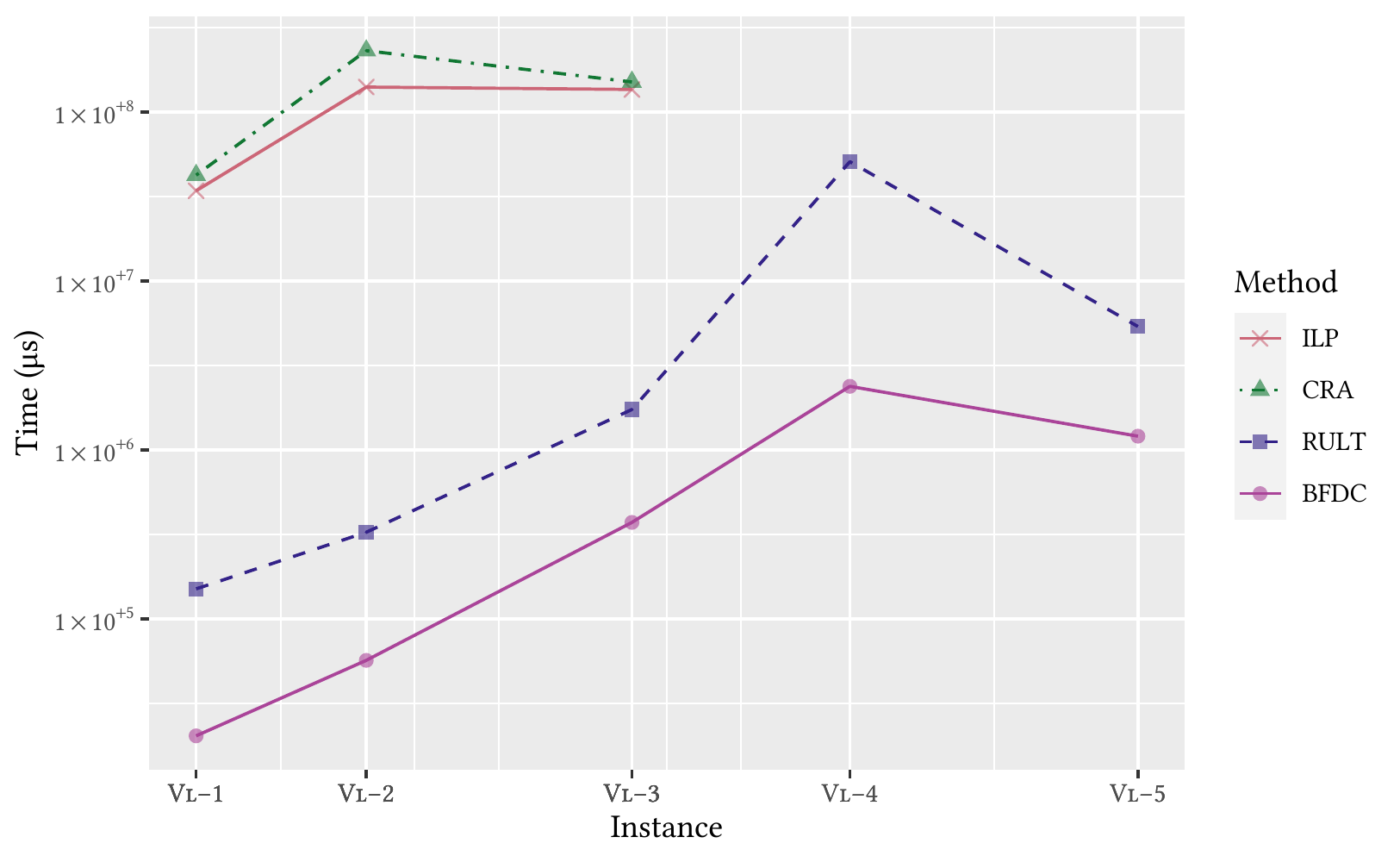}
    \caption{Average computation times for the \textsc{Vl} dataset.}
    \label{fig:vl}
\end{figure}

\section{Discussion} \label{sec:discussion}

 \citet{art:exp2} noted that a lot of what is reported in experimental papers are observations about the \textit{implementation} of an algorithm rather than the algorithm itself as a mathematical object. On the one hand, our study somewhat conforms to this trend. Despite being a well-known limitation of empirical studies, we hope to mitigate this by providing our code so that interested readers can inspect and even improve upon the implementations. On the other hand, some results documented in this paper are implementation-independent. For example: the reduced space complexity achieved by both RULT and BFDC compared to all other methods.

Let us now turn our attention to a broad analysis of the computational study. Table \ref{tab:summary} summarizes the results of the experiments from Section \ref{sec:exps}, with the exception of \textsc{Vl} instances. Columns \textit{Max. time (ms)}, \textit{Avg. time (ms)} and \textit{Std. time (ms)} report the maximum, average and standard deviation of the recorded execution times per method in milliseconds. Column \textit{Total time (s)} reports the total time required by each method to solve all of the instances, including eventual timeouts. Finally, column \textit{Timeouts (\%)} provides the percentage of runs for which the method timed out.

When considering Table \ref{tab:summary}, one must take into account the fact that ULT, CP, SCP and ILP are all more general than CRA, RULT, and BFDC. Hence, it should not come as a surprise that the latter algorithms outperform the former in almost all cases. Nevertheless, our experiments show that there are major differences in performance between algorithms when solving SDTPs and certain conclusions may appear counter-intuitive at first. For example, despite the polynomial worst-case asymptotic time complexity of ULT and KA, these algorithms exhibit poor general performance when solving SDTPs. Meanwhile, ILP demonstrated good performance for a problem that might have initially seemed more suitable for constraint programming.

\begin{table}[!ht]
\caption{Summary of results.}\label{tab:summary}
\centering
\begin{tabular}{lrrrrr}
\hline
Method & Max. time (ms) & Avg. time (ms) & Std. time (ms) & Total time (s) & Timeouts (\%)\\
\hline
ILP & 1947 & 350 & 462 & 7058 & 2.28\\
CP & 2017 & 1062 & 985 & 21406 & 33.39\\
SCP & 2002 & 875 & 954 & 17633 & 25.15\\
ULT & 1990 & 1476 & 812 & 29758 & 66.76\\
KAB & 2000 & 868 & 873 & 17491 & 28.02\\
KAJ & 2001 & 1065 & 1276 & 21478 & 30.26\\
CRA & 1895 & 138 & 317 & 2772 & 1.19\\
RULT & 152 & 6 & 13 & 115 & 0.00\\
BFDC & 28 & 1 & 2 & 18 & 0.00\\
\hline
\end{tabular}
\end{table}

Profiling the implementations of both KAB and KAJ showed that $\approx 90\%$ of their execution time was consistently spent building distance matrix $\delta$ and computing conflicts (lines 1-3 of Algorithm~\ref{alg:kra}). Less than $5\%$ of the total time was observed to be incurred by max-flow computations (line 6 of Algorithm~\ref{alg:kra}). This showcases how the bottleneck is the computation of APSPs and conflicts rather than the theoretically slower max-flow step.

Similarly, profiling the implementation of CRA showed that \textsc{Dijkstra} computations were responsible for up to 95\% of the processing time. This observation includes our lazy evaluation implementation of the distance matrix. When the full matrix is precomputed as originally described by \citet{art:cra}, the proportion of time spent on \textsc{Dijkstra} could grow even more extreme. In many cases, precomputation of the distance matrix was not possible within the imposed time limit. Furthermore, precomputing the distance matrix for large instances is simply not possible due to insufficient memory. In spite of these drawbacks, one advantage of CRA is that some instances can be solved quickly during its first stage (\textsc{BellmanFord}) when the initial solution is already feasible and there is therefore no time-point assignment $s_i$ which must undergo corrections.

The results detailed in Table \ref{tab:summary} further confirm those observed in Section \ref{sec:exps}. RULT and BFDC present the best performances overall, with computation times that are between two and three orders of magnitude shorter than all other methods. There is also no record of either of these algorithms timing out during our experiments. This performance can be explained by two factors. First, both RULT and BFDC focus on computing single-source shortest paths while ULT, KA and CRA consume a lot more computational resources solving APSPs. Second, and this comes as a direct consequence of the first factor, both RULT and BFDC have linear space complexity using only one-dimensional arrays of size $|T|$, thereby improving their cache locality and overall efficiency. Indeed, some instances could be solved almost entirely in cache by these two methods, while the quadratic space complexity of other methods made this far more unlikely.

Figure \ref{fig:cache} illustrates cache reference measurements for CRA, RULT and BFDC when solving the \textsc{Vl} instances. We focus on this dataset because it required the most algorithmic effort. Recording of cache reference events was performed using the \texttt{perf\_events} package from the Linux kernel \cite{web:perf}. The \textit{Full scale} row in the top half of the figure demonstrates how difficult it can sometimes be to compare the behavior of different approaches. This is why we have also included the \textit{Small scale} graphs below, which zoom into the \textit{Full scale} graphs in order to reveal further details concerning behavior of each method. These graphs make it clear how both RULT and BFDC can solve the large instances much quicker simply by using the cache more efficiently. In all experiments both RULT and BFDC required fewer total cache references than the number of cache misses by CRA.

\begin{figure}[H]
    \centering
    \includegraphics[width=0.95\linewidth]{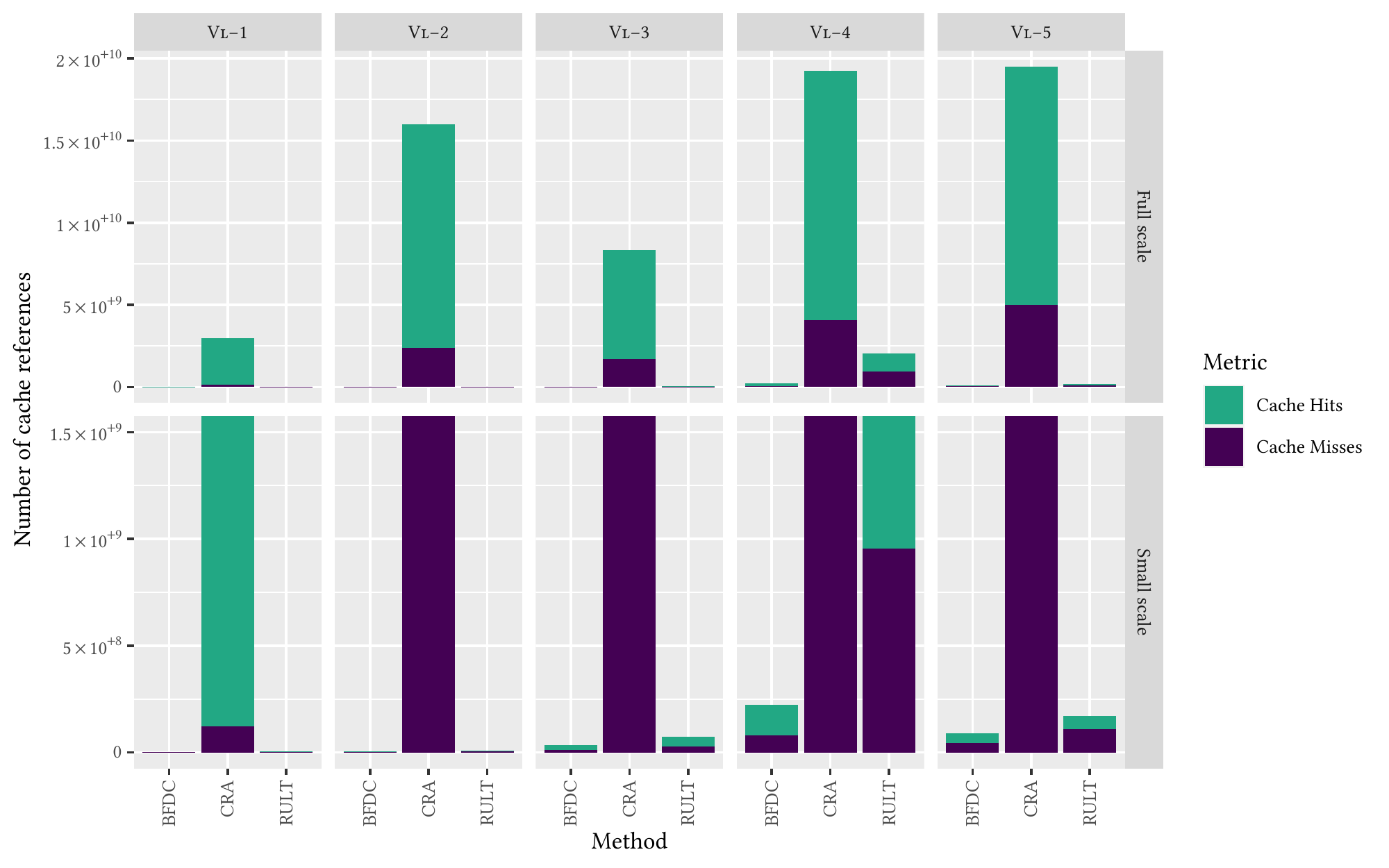}
    \caption{Cache references per method in the \textsc{Vl} dataset.}
    \label{fig:cache}
\end{figure}

The methods implemented in this paper also differ with respect to the type of solution produced. ULT, CRA, RULT and BFDC provide the earliest feasible solution at the end of their execution. However, ULT, RULT and BFDC can also return the latest feasible solution. \citet{art:cra} did not comment on whether their method could return the latest feasible solution, although it appears possible when computing solutions over graph $G_R$ and with some minor changes to operations (e.g., \textsc{UpdateAssignments}). By contrast, CP, SCP, ILP and KA are not guaranteed to return either the earliest or latest feasible solution. One could ensure finding either one of them by defining an appropriate objective function for the underlying model, but it is unclear how much this would impact their performance. For example, KA would require the solution of max-flows with arbitrary arc capacities rather than unitary capacities \cite{art:kumar-estp}.

While one could be tempted to conclude that BFDC and RULT should be the go-to methods when faced with SDTPs, this is not necessarily the conclusion we advocate for. Our advice is instead a little more nuanced. Given that BFDC performed the best for SDTPs on isolation, it represents the most sensible choice when evaluating, for example, the feasibility of interdependent vehicle routes. However, other problems which feature SDTPs may benefit from other algorithms to achieve the best performance. For restricted disjunctive temporal problems, \citet{art:cra} introduced a method which exploits CRA's structure to obtain a low time complexity. In theory, it is also possible to employ BFDC, but this would increase the asymptotic worst-case time complexity of the algorithm for RDTPs. Similarly, \citet{art:kumar-estp} showed that KA can be employed with minor changes to solve SDTPs where each domain is assigned an arbitrary preference weight and the goal is to find a feasible solution which maximizes the sum of the selected domains. In this problem context it is not clear how one could employ BFDC or RULT. On the other hand, we have shown empirically that KA experiences difficulty solving even medium-sized instances. Therefore, it may be worth considering further research on faster methods to solve these SDTPs with preferences.
 
\section{Conclusion}

Simple disjunctive temporal problems generalize simple temporal problems. They have a wide range of real-world applications where they typically arise as subproblems. Some examples of application domains include robot planning, workforce scheduling, logistics and management systems. SDTPs can also be used in decomposition methods to solve more general temporal constraint satisfaction problems. It is therefore of interest for both researchers and practitioners to understand the empirical performance of algorithms for solving SDTPs in addition to their theoretical time bounds. Unfortunately, the literature previously understood very little about these methods in practice.

To bridge this gap and bring theory and experimentation in these temporal problems closer together, we provided a large exploratory and empirical study concerning new and established algorithms for solving SDTPs. Our results indicate that theoretical worst-case time complexities are not necessarily indicative of the observed computation times of these algorithms. Moreover, we showed that the quadratic space complexity of previous algorithms comes with several drawbacks that limit their use in practice, regardless of their asymptotic time complexity. Indeed, for very large datasets, some methods were unable to solve an otherwise simple problem due to memory limitations, even when executed on modern computers. By contrast, algorithms which possess a lower space complexity albeit a higher time complexity solved very large instances within only a few seconds.

We hope that the results of this paper provide useful evidence for future researchers that helps them make informed decisions concerning the best algorithm for their application, thereby reducing reimplementation efforts. The code we have made publicly available will also help future research test whether our conclusions hold for other datasets. Finally, our implementations also provide some common ground for benchmarking new algorithms or speedup techniques for simple disjunctive temporal problems and their special cases.

 In terms of future research, one could consider performing a similar computational experiment for the restricted disjunctive temporal problem \cite{art:kra}. Instances could be derived from those introduced in this paper by adding Type 3 constraints. Another option is to investigate algorithms for variants of the SDTP with preferences associated with each domain of a time-point \cite{art:kumar-estp}. For example, a certain time-point may have greater preference to be executed in the morning rather than in the evening. Another possibility is to extract SDTP instances from real-world applications and verify whether the conclusions from our research remain valid for other graph structures or if better performing methods exist.

\begin{acks}

This research was supported by Internal Funds KU Leuven (IMP/20/021) and by the Flemish Government, Belgium under \textit{Onderzoeksprogramma Artificiële Intelligentie} (AI). Editorial consultation provided by Luke Connolly (KU Leuven).
\end{acks}

\bibliographystyle{ACM-Reference-Format}
\bibliography{biblio}


\begin{thebibliography}{29}

%%% the Web address.

\ifx \showCODEN    \undefined \def \showCODEN     #1{\unskip}     \fi
\ifx \showDOI      \undefined \def \showDOI       #1{#1}\fi
\ifx \showISBNx    \undefined \def \showISBNx     #1{\unskip}     \fi
\ifx \showISBNxiii \undefined \def \showISBNxiii  #1{\unskip}     \fi
\ifx \showISSN     \undefined \def \showISSN      #1{\unskip}     \fi
\ifx \showLCCN     \undefined \def \showLCCN      #1{\unskip}     \fi
\ifx \shownote     \undefined \def \shownote      #1{#1}          \fi
\ifx \showarticletitle \undefined \def \showarticletitle #1{#1}   \fi
\ifx \showURL      \undefined \def \showURL       {\relax}        \fi
% invisible to TeX
\providecommand\bibfield[2]{#2}
\providecommand\bibinfo[2]{#2}
\providecommand\natexlab[1]{#1}
\providecommand\showeprint[2][]{arXiv:#2}

\bibitem[Ahuja et~al\mbox{.}(1993)]%
        {book:networks}
\bibfield{author}{\bibinfo{person}{Ravindra~K. Ahuja},
  \bibinfo{person}{Thomas~L. Magnanti}, {and} \bibinfo{person}{James~B.
  Orlin}.} \bibinfo{year}{1993}\natexlab{}.
\newblock \bibinfo{booktitle}{\emph{Network flows: theory, algorithms, and
  applications}}.
\newblock \bibinfo{publisher}{Prentice-Hall, Inc.}
\newblock


\bibitem[Cesta and Oddi(1996)]%
        {art:stp-bf-inc}
\bibfield{author}{\bibinfo{person}{Amedeo Cesta} {and} \bibinfo{person}{Angelo
  Oddi}.} \bibinfo{year}{1996}\natexlab{}.
\newblock \showarticletitle{Gaining efficiency and flexibility in the simple
  temporal problem}. In \bibinfo{booktitle}{\emph{Proceedings Third
  International Workshop on Temporal Representation and Reasoning (TIME'96)}}.
  IEEE, \bibinfo{pages}{45--50}.
\newblock


\bibitem[Cherkassky et~al\mbox{.}(2010)]%
        {art:sp-fp}
\bibfield{author}{\bibinfo{person}{Boris~V. Cherkassky},
  \bibinfo{person}{Loukas Georgiadis}, \bibinfo{person}{Andrew~V. Goldberg},
  \bibinfo{person}{Robert~E. Tarjan}, {and} \bibinfo{person}{Renato~F.
  Werneck}.} \bibinfo{year}{2010}\natexlab{}.
\newblock \showarticletitle{Shortest-Path Feasibility Algorithms: An
  Experimental Evaluation}.
\newblock \bibinfo{journal}{\emph{ACM J. Exp. Algorithmics}}
  \bibinfo{volume}{14}, Article \bibinfo{articleno}{7} (\bibinfo{date}{jan}
  \bibinfo{year}{2010}), \bibinfo{numpages}{37}~pages.
\newblock
\showISSN{1084-6654}


\bibitem[Cherkassky et~al\mbox{.}(1996)]%
        {art:splib}
\bibfield{author}{\bibinfo{person}{Boris~V. Cherkassky},
  \bibinfo{person}{Andrew~V. Goldberg}, {and} \bibinfo{person}{Tomasz Radzik}.}
  \bibinfo{year}{1996}\natexlab{}.
\newblock \showarticletitle{Shortest paths algorithms: Theory and experimental
  evaluation}.
\newblock \bibinfo{journal}{\emph{Mathematical programming}}
  \bibinfo{volume}{73}, \bibinfo{number}{2} (\bibinfo{year}{1996}),
  \bibinfo{pages}{129--174}.
\newblock


\bibitem[Christiaens and Berghe(2020)]%
        {art:sisrs}
\bibfield{author}{\bibinfo{person}{Jan Christiaens} {and}
  \bibinfo{person}{Greet~Vanden Berghe}.} \bibinfo{year}{2020}\natexlab{}.
\newblock \showarticletitle{Slack Induction by String Removals for Vehicle
  Routing Problems}.
\newblock \bibinfo{journal}{\emph{Transportation Science}}
  \bibinfo{volume}{54}, \bibinfo{number}{2} (\bibinfo{year}{2020}),
  \bibinfo{pages}{417--433}.
\newblock


\bibitem[Comin and Rizzi(2018)]%
        {art:cra}
\bibfield{author}{\bibinfo{person}{Carlo Comin} {and} \bibinfo{person}{Romeo
  Rizzi}.} \bibinfo{year}{2018}\natexlab{}.
\newblock \showarticletitle{{On Restricted Disjunctive Temporal Problems:
  Faster Algorithms and Tractability Frontier}}. In
  \bibinfo{booktitle}{\emph{25th International Symposium on Temporal
  Representation and Reasoning (TIME 2018)}} \emph{(\bibinfo{series}{Leibniz
  International Proceedings in Informatics (LIPIcs)},
  Vol.~\bibinfo{volume}{120})}, \bibfield{editor}{\bibinfo{person}{Natasha
  Alechina}, \bibinfo{person}{Kjetil N{\o}rv{\aa}g}, {and}
  \bibinfo{person}{Wojciech Penczek}} (Eds.). \bibinfo{publisher}{Schloss
  Dagstuhl--Leibniz-Zentrum fuer Informatik}, \bibinfo{address}{Dagstuhl,
  Germany}, \bibinfo{pages}{10:1--10:20}.
\newblock


\bibitem[Cormen et~al\mbox{.}(2009)]%
        {book:cormen}
\bibfield{author}{\bibinfo{person}{Thomas~H. Cormen},
  \bibinfo{person}{Charles~E. Leiserson}, \bibinfo{person}{Ronald~L. Rivest},
  {and} \bibinfo{person}{Clifford Stein}.} \bibinfo{year}{2009}\natexlab{}.
\newblock \bibinfo{booktitle}{\emph{Introduction to algorithms}}.
\newblock \bibinfo{publisher}{MIT press}.
\newblock


\bibitem[Dechter et~al\mbox{.}(1991)]%
        {art:stp}
\bibfield{author}{\bibinfo{person}{Rina Dechter}, \bibinfo{person}{Itay Meiri},
  {and} \bibinfo{person}{Judea Pearl}.} \bibinfo{year}{1991}\natexlab{}.
\newblock \showarticletitle{Temporal constraint networks}.
\newblock \bibinfo{journal}{\emph{Artificial Intelligence}}
  \bibinfo{volume}{49}, \bibinfo{number}{1} (\bibinfo{year}{1991}),
  \bibinfo{pages}{61 -- 95}.
\newblock


\bibitem[Fredman and Tarjan(1987)]%
        {art:fib-heaps}
\bibfield{author}{\bibinfo{person}{Michael~L. Fredman} {and}
  \bibinfo{person}{Robert~Endre Tarjan}.} \bibinfo{year}{1987}\natexlab{}.
\newblock \showarticletitle{Fibonacci Heaps and Their Uses in Improved Network
  Optimization Algorithms}.
\newblock \bibinfo{journal}{\emph{J. ACM}} \bibinfo{volume}{34},
  \bibinfo{number}{3} (\bibinfo{date}{jul} \bibinfo{year}{1987}),
  \bibinfo{pages}{596–615}.
\newblock


\bibitem[Goldberg and Tarjan(1988)]%
        {art:max-flow}
\bibfield{author}{\bibinfo{person}{Andrew~V. Goldberg} {and}
  \bibinfo{person}{Robert~E. Tarjan}.} \bibinfo{year}{1988}\natexlab{}.
\newblock \showarticletitle{A new approach to the maximum-flow problem}.
\newblock \bibinfo{journal}{\emph{Journal of the ACM (JACM)}}
  \bibinfo{volume}{35}, \bibinfo{number}{4} (\bibinfo{year}{1988}),
  \bibinfo{pages}{921--940}.
\newblock


\bibitem[Hojabri et~al\mbox{.}(2018)]%
        {art:vrpms}
\bibfield{author}{\bibinfo{person}{Hossein Hojabri}, \bibinfo{person}{Michel
  Gendreau}, \bibinfo{person}{Jean-Yves Potvin}, {and}
  \bibinfo{person}{Louis-Martin Rousseau}.} \bibinfo{year}{2018}\natexlab{}.
\newblock \showarticletitle{Large neighborhood search with constraint
  programming for a vehicle routing problem with synchronization constraints}.
\newblock \bibinfo{journal}{\emph{Computers \& Operations Research}}
  \bibinfo{volume}{92} (\bibinfo{year}{2018}), \bibinfo{pages}{87 -- 97}.
\newblock


\bibitem[Hunsberger and Posenato(2021)]%
        {art:time21}
\bibfield{author}{\bibinfo{person}{Luke Hunsberger} {and}
  \bibinfo{person}{Roberto Posenato}.} \bibinfo{year}{2021}\natexlab{}.
\newblock \showarticletitle{{Simple Temporal Networks: A Practical Foundation
  for Temporal Representation and Reasoning}}. In
  \bibinfo{booktitle}{\emph{28th International Symposium on Temporal
  Representation and Reasoning (TIME 2021)}} \emph{(\bibinfo{series}{Leibniz
  International Proceedings in Informatics (LIPIcs)},
  Vol.~\bibinfo{volume}{206})}, \bibfield{editor}{\bibinfo{person}{Carlo
  Combi}, \bibinfo{person}{Johann Eder}, {and} \bibinfo{person}{Mark Reynolds}}
  (Eds.). \bibinfo{publisher}{Schloss Dagstuhl -- Leibniz-Zentrum f{\"u}r
  Informatik}, \bibinfo{address}{Dagstuhl, Germany}, \bibinfo{pages}{1:1--1:5}.
\newblock
\showISBNx{978-3-95977-206-8}
\showISSN{1868-8969}


\bibitem[Johnson(2002)]%
        {art:exp2}
\bibfield{author}{\bibinfo{person}{David~S. Johnson}.}
  \bibinfo{year}{2002}\natexlab{}.
\newblock \showarticletitle{A Theoretician’s Guide to the Experimental
  Analysis of Algorithms}.
\newblock \bibinfo{journal}{\emph{DIMACS Series in Discrete Mathematics and
  Theoretical Computer Science}}  \bibinfo{volume}{59} (\bibinfo{year}{2002}).
\newblock


\bibitem[Kumar(2004)]%
        {art:kumar-estp}
\bibfield{author}{\bibinfo{person}{T.K.~Satish Kumar}.}
  \bibinfo{year}{2004}\natexlab{}.
\newblock \showarticletitle{A polynomial-time algorithm for simple temporal
  problems with piecewise constant domain preference functions}. In
  \bibinfo{booktitle}{\emph{AAAI}}. \bibinfo{pages}{67--72}.
\newblock


\bibitem[Kumar(2005)]%
        {art:kra}
\bibfield{author}{\bibinfo{person}{T.K.~Satish Kumar}.}
  \bibinfo{year}{2005}\natexlab{}.
\newblock \showarticletitle{On the Tractability of Restricted Disjunctive
  Temporal Problems.}. In \bibinfo{booktitle}{\emph{ICAPS}}.
  \bibinfo{pages}{110--119}.
\newblock


\bibitem[Kumar et~al\mbox{.}(2013)]%
        {art:stpts}
\bibfield{author}{\bibinfo{person}{T.K.~Satish Kumar},
  \bibinfo{person}{Marcello Cirillo}, {and} \bibinfo{person}{Sven Koenig}.}
  \bibinfo{year}{2013}\natexlab{}.
\newblock \showarticletitle{Simple temporal problems with taboo regions}. In
  \bibinfo{booktitle}{\emph{Twenty-Seventh AAAI Conference on Artificial
  Intelligence}}.
\newblock


\bibitem[Larkin et~al\mbox{.}(2014)]%
        {art:heaps}
\bibfield{author}{\bibinfo{person}{Daniel~H. Larkin},
  \bibinfo{person}{Siddhartha Sen}, {and} \bibinfo{person}{Robert~E. Tarjan}.}
  \bibinfo{year}{2014}\natexlab{}.
\newblock \showarticletitle{A back-to-basics empirical study of priority
  queues}. In \bibinfo{booktitle}{\emph{2014 Proceedings of the Sixteenth
  Workshop on Algorithm Engineering and Experiments (ALENEX)}}. SIAM,
  \bibinfo{pages}{61--72}.
\newblock


\bibitem[Masson et~al\mbox{.}(2014)]%
        {art:darpt}
\bibfield{author}{\bibinfo{person}{Renaud Masson}, \bibinfo{person}{Fabien
  Lehu\'ed\'e}, {and} \bibinfo{person}{Olivier P\'eton}.}
  \bibinfo{year}{2014}\natexlab{}.
\newblock \showarticletitle{The Dial-A-Ride Problem with Transfers}.
\newblock \bibinfo{journal}{\emph{Computers \& Operations Research}}
  \bibinfo{volume}{41} (\bibinfo{year}{2014}), \bibinfo{pages}{12 -- 23}.
\newblock


\bibitem[McGeoch(1996)]%
        {art:exp1}
\bibfield{author}{\bibinfo{person}{Catherine~C McGeoch}.}
  \bibinfo{year}{1996}\natexlab{}.
\newblock \showarticletitle{Toward an experimental method for algorithm
  simulation}.
\newblock \bibinfo{journal}{\emph{INFORMS Journal on Computing}}
  \bibinfo{volume}{8}, \bibinfo{number}{1} (\bibinfo{year}{1996}),
  \bibinfo{pages}{1--15}.
\newblock


\bibitem[Mitrovi{\'c}-Mini{\'c} and Laporte(2006)]%
        {art:pdpt}
\bibfield{author}{\bibinfo{person}{Sne{\v{z}}ana Mitrovi{\'c}-Mini{\'c}} {and}
  \bibinfo{person}{Gilbert Laporte}.} \bibinfo{year}{2006}\natexlab{}.
\newblock \showarticletitle{The pickup and delivery problem with time windows
  and transshipment}.
\newblock \bibinfo{journal}{\emph{INFOR: Information Systems and Operational
  Research}} \bibinfo{volume}{44}, \bibinfo{number}{3} (\bibinfo{year}{2006}),
  \bibinfo{pages}{217--227}.
\newblock


\bibitem[Moret and Shapiro(2001)]%
        {art:shapiro}
\bibfield{author}{\bibinfo{person}{Bernard~ME. Moret} {and}
  \bibinfo{person}{Henry~D. Shapiro}.} \bibinfo{year}{2001}\natexlab{}.
\newblock \showarticletitle{Algorithms and experiments: The new (and the old)
  methodology}.
\newblock \bibinfo{journal}{\emph{Journal of Universal Computer Science}}
  \bibinfo{volume}{7} (\bibinfo{year}{2001}), \bibinfo{pages}{434--446}.
\newblock


\bibitem[Pralet and Verfaillie(2012)]%
        {art:td-stp}
\bibfield{author}{\bibinfo{person}{C{\'e}dric Pralet} {and}
  \bibinfo{person}{G{\'e}rard Verfaillie}.} \bibinfo{year}{2012}\natexlab{}.
\newblock \showarticletitle{Time-dependent simple temporal networks}. In
  \bibinfo{booktitle}{\emph{International Conference on Principles and Practice
  of Constraint Programming}}. Springer, \bibinfo{pages}{608--623}.
\newblock


\bibitem[Reinelt(1991)]%
        {art:tsplib}
\bibfield{author}{\bibinfo{person}{Gerhard Reinelt}.}
  \bibinfo{year}{1991}\natexlab{}.
\newblock \showarticletitle{TSPLIB—A traveling salesman problem library}.
\newblock \bibinfo{journal}{\emph{ORSA journal on computing}}
  \bibinfo{volume}{3}, \bibinfo{number}{4} (\bibinfo{year}{1991}),
  \bibinfo{pages}{376--384}.
\newblock


\bibitem[Sanders(2001)]%
        {art:sequence-heaps}
\bibfield{author}{\bibinfo{person}{Peter Sanders}.}
  \bibinfo{year}{2001}\natexlab{}.
\newblock \showarticletitle{Fast Priority Queues for Cached Memory}.
\newblock \bibinfo{journal}{\emph{ACM J. Exp. Algorithmics}}
  \bibinfo{volume}{5} (\bibinfo{date}{dec} \bibinfo{year}{2001}),
  \bibinfo{numpages}{25}~pages.
\newblock
\showISSN{1084-6654}


\bibitem[Sarasola and Doerner(2020)]%
        {art:delsynch}
\bibfield{author}{\bibinfo{person}{Briseida Sarasola} {and}
  \bibinfo{person}{Karl~F. Doerner}.} \bibinfo{year}{2020}\natexlab{}.
\newblock \showarticletitle{Adaptive large neighborhood search for the vehicle
  routing problem with synchronization constraints at the delivery location}.
\newblock \bibinfo{journal}{\emph{Networks}} \bibinfo{volume}{75},
  \bibinfo{number}{1} (\bibinfo{year}{2020}), \bibinfo{pages}{64--85}.
\newblock


\bibitem[Sartori et~al\mbox{.}(2022)]%
        {art:ssvb-1}
\bibfield{author}{\bibinfo{person}{Carlo~S. Sartori}, \bibinfo{person}{Pieter
  Smet}, {and} \bibinfo{person}{Greet {Vanden Berghe}}.}
  \bibinfo{year}{2022}\natexlab{}.
\newblock \showarticletitle{{Scheduling truck drivers with interdependent
  routes under European Union regulations}}.
\newblock \bibinfo{journal}{\emph{European Journal of Operational Research}}
  \bibinfo{volume}{298}, \bibinfo{number}{1} (\bibinfo{year}{2022}),
  \bibinfo{pages}{76--88}.
\newblock


\bibitem[Schwalb and Dechter(1997)]%
        {art:ult}
\bibfield{author}{\bibinfo{person}{Eddie Schwalb} {and} \bibinfo{person}{Rina
  Dechter}.} \bibinfo{year}{1997}\natexlab{}.
\newblock \showarticletitle{Processing disjunctions in temporal constraint
  networks}.
\newblock \bibinfo{journal}{\emph{Artificial Intelligence}}
  \bibinfo{volume}{93}, \bibinfo{number}{1} (\bibinfo{year}{1997}),
  \bibinfo{pages}{29--61}.
\newblock


\bibitem[Stergiou and Koubarakis(2000)]%
        {art:dtps}
\bibfield{author}{\bibinfo{person}{Kostas Stergiou} {and}
  \bibinfo{person}{Manolis Koubarakis}.} \bibinfo{year}{2000}\natexlab{}.
\newblock \showarticletitle{Backtracking algorithms for disjunctions of
  temporal constraints}.
\newblock \bibinfo{journal}{\emph{Artificial Intelligence}}
  \bibinfo{volume}{120}, \bibinfo{number}{1} (\bibinfo{year}{2000}),
  \bibinfo{pages}{81--117}.
\newblock


\bibitem[Weaver(2013)]%
        {web:perf}
\bibfield{author}{\bibinfo{person}{Vince Weaver}.}
  \bibinfo{year}{2013}\natexlab{}.
\newblock \bibinfo{title}{{The Unofficial Linux Perf Events Web-Page}}.
\newblock
\newblock
\newblock
\shownote{Available online at
  \url{https://web.eece.maine.edu/~vweaver/projects/perf_events/}. Last access
  on: 2022-10-17}.


\end{thebibliography}

\appendix

\section{Instance generation} \label{ap:instances}

In this appendix we provide additional details concerning how some of the instances used during our computational study were generated. We describe in further detail the shortest path instances and the vehicle routing instances.

\subsection{Shortest path instances}

In what follows we expand upon the details already given per \textsc{Sp} dataset in Section \ref{sec:exps}.

\begin{enumerate}
    \item \textsc{Rand}: generates graph $G$ using \textsc{Sprand}, introduced in SPLib \cite{art:splib}. Nodes and arcs are all randomly created. Constraints $C_2$ are created by computing the shortest path $\tau$ from a dummy source node $z$ to every $i \in V$. A number of time-points is randomly added to set $T_D$ until the desired size $|T_D|$ is achieved. For each $j \in T_D$, we generate a random number $\kappa_j$ from the uniform distribution $\mathcal{U}[1,K]$. The $\kappa_j$-th domain of $j$ is then defined as $d^{\kappa_j}_j = [s^0_i - \varphi_1,s^0_i+\varphi_2]$ where $s^0_i=-\tau_{i}$ and $\varphi_1,\varphi_2 \in \mathcal{U}[0,2000]$. Once $d^{\kappa_j}_j$ has been defined, the remaining domains are generated as follows. First, domains $[l^1_j,u^1_j],\dots,[l^{\kappa_j-1}_j,u^{\kappa_j-1}_j]$ are created by working backwards from $[l^{\kappa_j}_j,u^{\kappa_j}_j]$. Then, domain $[l^{\kappa_j-1}_j,u^{\kappa_j-1}_j]$ is defined via $u^{\kappa_j-1}_j=l^{\kappa_j}_j - a - 1$ and $l^{\kappa_j-1}_j=u^{\kappa_j-1}_j - b_j$ where $a \in \mathcal{U}[0,200]$ and $b_j=u^{\kappa_j}_j-l^{\kappa_j}_j$. This is repeated recursively until $[l^1_j,u^1_j]$. Similarly, for the last domains $[l^{\kappa_j+1}_j,u^{\kappa_j+1}_j],\dots,[l^K_j,u^K_j]$ we generate $[l^{\kappa_j+1}_j,u^{\kappa_j+1}_j]$ by setting $l^{\kappa_j+1}_j = u^{\kappa_j}_j + a + 1$ and $u^{\kappa_j+1}_j = l^{\kappa_j+1}_j + b_j$. Again, this is repeated recursively until the $K$-th domain of $j$ has been created. We only accept instances where at least 60\% of the entries $s_i$ in the earliest feasible solution $s$ of the SDTP belong to a domain $[l^c_i,u^c_i],\ c > 1$. This means that feasibility is not ensured by simply assigning the first domain to every time-point.
    \item \textsc{Grid}: generates graph $G$ using \textsc{Spgrid}, also introduced in SPLib \cite{art:splib}. Nodes are generated in a grid format with $X$ layers and $Y$ nodes per layer. Arcs connect nodes within the same layer and to those in subsequent layers. Additional arcs may be included between nodes in different layers. Similar to \citet{art:sp-fp}, we fix $Y=16$ in all graphs generated with \textsc{Spgrid} and vary $X$ as required to create $|T|=16\cdot X$ time-points for the SDTP instance. Varying parameter $Y$ did not change the results significantly. $C_2$ constraints are generated in the same way as for \textsc{Rand}.
    \item \textsc{Seq}: generates graph $G$ using the tailored generator \textsc{Spseq}. Nodes are generated at random, similarly to \textsc{Sprand}. A path connecting all nodes with $|V|-1$ arcs is created, with each arc having weight $w_{ij}=1$. Afterwards, the remaining $|A|-|V|+1$ arcs are created at random with greater weights selected uniformly from $\mathcal{U}[500,20000]$. This creates a known shortest path which some methods may experience difficulty finding. $C_2$ constraints are generated in the same way as for \textsc{Rand}.
    \item \textsc{Late}: generates graph $G$ using both \textsc{Sprand} and \textsc{Spseq}. However, $C_2$ constraints are created differently. First, initial domains $d^0_i = [s^0_i - \varphi_1,s^0_i+\varphi_2]$ are generated as in the \textsc{Rand} dataset for all $i \in T$. Afterwards, we select a subset of time-points to include in $T_D$. For each $j \in T_D$ we set $D_j=\{d^0_j\}$ and continue to iterate as follows. First, a time-point $j \in T_D,\ |D_j| < K$ is randomly selected. We then create a new domain $[l^c_j,u^c_j]$ for $j$ as $l^c_j=-\tau_{j}+5$ and $u^c_j=l^c_j+1$. The new domain is appended to $D_j$ and earlier domains are shifted if $u^{c-1}_j \geq l^c_j$. Once this has been completed, we update the weights in $G$ for arcs $(z,i,w_{zi})$ to $w_{zi}=-l^c_j$ and recompute distances $\tau$. This effectively updates the solutions and adapts the domains that will be generated in subsequent iterations. The next iteration then begins by selecting another $j \in T_D$ and appending a new domain to $D_j$ until $|D_j|=K\ \forall\ j \in T_D$. This procedure enables us to enforce feasible solutions to belong to \textit{later} domains in each time-point $j \in T_D$, potentially meaning that some algorithms will have to run longer if they incrementally build up from the first domain of each time-point or if they have to tighten the global boundaries. We only generate \textsc{Late} instances for which at least 60\% of the entries $s_i$ in the earliest feasible solution $s$ belong to the last domain of their respective time-point.
\end{enumerate}

\subsection{Vehicle routing instances}

Vehicle routing instances are extracted from those introduced for Vehicle Routing Problems with Multiple Synchronization constraints (VRPMS) \cite{art:vrpms}. We solve VRPMS instances using heuristic Slack Induction by String Removals \cite{art:sisrs}. Each new feasible VRPMS solution is saved as a graph which contains multiple chains of nodes (routes) connected via synchronization arcs. Duration arcs are included to create cycles and produce more challenging instances. These problems contain one time window per customer location and can therefore be interpreted as a simple temporal problem~\cite{art:ssvb-1}.

To create SDTP instances, we transform the aforementioned graphs into an SDTN network $N=(T,C)$. The nodes (customers in VRPMS) are transformed into time-points together with the depots from where vehicles depart. $C_1$ constraints are created from the travel times between locations in a route as well as interdependency constraints between two routes. Maximum route duration constraints also feature in $C_1$. For $C_2$, we split the time windows of the VRPMS nodes by simply selecting a gap size $a$ between the domains and then splitting the original time window into $K$ disjunctive domains. For two consecutive intervals $[l^c_i,u^c_i]$ and $[l^{c+1}_i,u^{c+1}_i]$, we set $l^{c+1}=u^c_i+a$. In this way, we maintain as many of the original characteristics of the VRPMS instance as possible while ensuring that we still create a complete SDTP instance from a real application.

\end{document}